\documentclass[10pt,a4paper]{article}

\usepackage[margin=1.2in]{geometry}
\usepackage{amsmath,amssymb,amsthm}
\usepackage[T1]{fontenc}
\usepackage[latin1]{inputenc}
\usepackage[tiny]{titlesec}
\usepackage{hyperref}
\usepackage[numbers]{natbib}
\usepackage{booktabs}

\newcommand{\U}{\mathcal{U}}

\newcommand{\M}{\mathcal{M}}
\newcommand{\X}{\mathcal{S}}

\renewcommand{\d}{\mathfrak{d}_{\textsc{J}_1}}

\newcommand{\N}{\mathbb{N}}
\newcommand{\E}{\mathcal{E}}
\renewcommand{\SS}{\mathbb{S}}

\newcommand{\D}{\mathcal{D}}

\newcommand{\B}{\mathcal{B}}
\newcommand{\R}{\mathbb{R}}

\newcommand{\I}{1{\hskip -2.5 pt}\hbox{I}}
\newcommand{\II}{\mathbf{I}}

\newcommand{\ACKNO}[1]{\noindent\textbf{Acknowledgments.} #1}
\newcommand{\DA}[1]{\noindent\textbf{Data availability.} #1}
\newcommand{\COI}[1]{\noindent\textbf{Conflicts of Interest.} #1}

\theoremstyle{plain}
\newtheorem{Thm}{Theorem}[section]
\newtheorem{Lem}[Thm]{Lemma}
\newtheorem{Prop}[Thm]{Proposition}
\newtheorem*{Prop*}{Proposition}
\newtheorem{Cor}[Thm]{Corollary}

\theoremstyle{definition}
\newtheorem{Def}[Thm]{Definition}
\newtheorem{Eg}[Thm]{Example}

\theoremstyle{remark}
\newtheorem{Rem}[Thm]{Remark}
\def \follmer {F\"{o}llmer }
\def \ito {It\^{o} }

\def \cadlag {c\`adl\`ag }

\def \closed {generic} 
\def\mint{\mathop{\,\rlap{-}\!\!\int}\nolimits}

\title{Model-free portfolio allocation in continuous-time}

\author{Henry CHIU \footnote{School of Mathematics, University of Birmingham.
    h.chiu@bham.ac.uk}
  }

\date{}
\begin{document}

\maketitle

\begin{abstract}
We present a non-probabilistic, path-by-path framework for studying path-dependent (i.e., where weight is a functional of time and historical time-series), long-only portfolio allocation in continuous-time based on \cite{CC3}, where the fundamental concept of self-financing was introduced, independent of any integration theory. In this article, we extend this concept to a portfolio allocation strategy and characterize it by a path-dependent partial differential equation. We derive the general explicit solution that describes the evolution of wealth in generic markets, including price paths that may not evolve continuously or exhibit variation of any order. Explicit solution examples are provided.

As an application of our continuous-time, path-dependent framework, we extend an aggregating algorithm of Vovk \cite{VV} and the universal algorithm of Cover \cite{TC} to continuous-time algorithms that combine multiple strategies into a single strategy. These continuous-time (meta) algorithms take multiple strategies as input (which may themselves be generated by other algorithms) and track the wealth generated by the best individual strategy and the best convex combination of strategies, with tracking error bounds in log wealth of order $O(1)$ and $O(\ln t)$, respectively. This work provides the first extension of Cover's theorem \cite[Thm 6.1]{TC} to a continuous-time, model-free setting, recovering his celebrated error bound.
\end{abstract}

\noindent\textbf{Keywords:}
Model-free; path-dependent portfolio allocation; continuous-time; universal portfolio.

\tableofcontents
\newpage
\section{Introduction}\noindent

Portfolio allocation in continuous-time was developed using probability concepts \cite[Merton (1969)]{RM}. We also reference here several seminal works: the introduction of continuous-time Kelly's criterion by Thorp \cite{ET}, Cover's universal portfolio by Jamshidian \cite{FJ} (see also \cite[Ishijima (2002)]{HI} for a connection between Kelly's criterion and Cover's portfolio) and stochastic portfolio theory of Fernholz \cite{ERF}. 

In these works and much of the subsequent research, it is assumed that the evolution of stock prices follows a probabilistic model (or a collection of models). These models produce statements or results that are only valid \emph {almost surely}. For example, a relative arbitrage strategy that almost surely outperforms a benchmark. In other words, it is \emph{analytically possible} for such a strategy to fall short of a benchmark, should a scenario of probability zero come to realisation. For example, in a diffusion model, such scenarios may include, but not limited to, price paths that exhibit discontinuities. Consequently, these model-based approaches may be viewed as less reliable by stakeholders.

Over the last decade, there has been a growing body of research focused on developing model-free or model-agnostic, path-by-path approaches to continuous-time portfolio problems. These approaches do not rely on probabilistic models but instead only use quantities and time-series data that are observable and computable path-by-path. See, for example \cite{AS1}, \cite{AS2}, \cite{CSW}, \cite{KK}, \cite{All} and \cite{HS}, which focus on continuous price paths of finite quadratic variations. These approaches depend on the "choice" of a path integration concept to define the value of a "self-financing" portfolio. While this method may be intuitive, it may also be imprecise, a concern that Lyons had already raised in an earlier work \cite[\S 2.2]{TL}. As a demonstration, we present in Example \ref{Eg:counter} an integrand whose left-Riemann sums exist (i.e. its path integral may be defined) but does not constitute a self-financing trading strategy \cite{CC3}.

In this article, we present a model-free, path-by-path approach for studying path-dependent, long-only portfolio allocation in continuous-time, based on the recent framework introduced in \cite{CC3}. Unlike other approaches, our definition of a self-financing trading strategy does not rely on a specific integration theory or the existence of finite quadratic variation. We first extend the concept of self-financing to a portfolio allocation strategy (a functional taking values in a simplex) and associate this strategy with a path-dependent partial differential equation (PPDE). We demonstrate that a portfolio allocation strategy is self-financing if and only if its associated PPDE admits a solution, which characterizes the evolution of wealth associated with the strategy. We derive the general explicit solution for any generic domain, including price paths that do not evolve continuously or exhibit variation of any order, and provide examples of concrete solutions. In particular, we show that relative arbitrage does not exist on any generic domain; that is, if two portfolio allocation strategies are self-financing, neither strategy's wealth may dominate the other's in finite time.

As an application of our continuous-time, path-dependent framework, we extend two discrete-time machine learning algorithms to continuous-time meta-learning algorithms. These algorithms take multiple strategies as input (which may themselves be generated by other algorithms) and track the wealth generated by the best individual strategy and the best convex combination of strategies, respectively. Specifically, we show that their tracking errors in log wealth are bounded by $O(1)$ and $O(\ln t)$, respectively.  

The first algorithm operates on any generic domain and is based on an  aggregating algorithm of Vovk \cite{VV} (also known as Exponential weights or the Laissez-faire algorithm in the investment context \cite[\S 1.2]{YK}), which belongs to the class of online learning with expert advice algorithms. Using this algorithm, we construct in Example \ref{eg:insured} an explicit \emph{path-dependent} allocation strategy that does not necessarily evolve continuously or exhibit variation of any order. The second algorithm operates on paths with finite quadratic variations with bounded relative jumps and is based on Cover's original algorithm \cite{TC}, which only applies to constant rebalanced portfolios. We extend this algorithm to the convex hull generated by multiple strategies and prove Cover's main theorem \cite[Thm. 6.1]{TC} in this context.

Since the seminal publication of Cover's \emph{universal portfolios} as the lead article in the inaugural issue of \emph{Mathematical Finance}, Cover's portfolio has been studied and generalised in various directions in continuous-time (see, for instance, \cite{FJ}, \cite{HI}, \cite{CSW}, and \cite{All}). Despite these efforts, Cover's main theorem regarding his celebrated error bound \cite[Thm. 6.1]{TC}, which was formulated in a non-probabilistic framework, remains unproven in the continuous-time setting without probability. We hope this article addresses this gap.

\subsection{Literature positioning}\label{sec:RL}

\paragraph{Model-free approaches}
To clarify our positioning within model-free approaches, we highlight three fundamental distinctions. First, whereas \follmer\!'s calculus and rough path theory are primarily pathwise integration frameworks, causal functional calculus \cite{CC2} provides a differential calculus on path space. Crucially, its differentiability class encompasses pathwise integrals themselves---notably \follmer\!'s pathwise integral, which lies outside the corresponding differentiability class of functional \ito calculus \cite{CF}. Second, this foundation operates on generic domains, accommodating price paths regardless of variation index and avoiding unverifiable constraints, such as the algebraic postulates of Chen's relations or the uniform convergence of iterated cross-sums.  Third, we do not define the concept of self-financing via an integration theory, which Example \ref{Eg:counter} shows can lead to fundamental inconsistencies. Instead, we use purely operational, local functional conditions, ensuring the value functional is exactly the PnL. Finally, our framework includes no-arbitrage results.

\paragraph{Cover's universal portfolios}
Cover's (1991) original theorem \cite[Thm. 6.1]{TC}, which established an error bound of $O(\ln t)$ in log-wealths, applies to constant rebalanced portfolios in discrete time without probability. In continuous time, Cover's theorem was first proven for continuous  paths by Jamshidian (1992) \cite[Thm. 2.4]{FJ} within a probabilistic framework. 

The extension of Cover's universal portfolio to continuous paths without probability was first studied by Cuchiero, Schachermayer \& Wong (2019) \cite{CSW}, who considered tracking a non-parametric family of Markovian allocations using \follmer\!'s calculus. Allan, Cuchiero, Liu, \& Pr\"{o}mel (2023) \cite{All} extended this approach to more general allocations using rough path theory. Subsequent to our initial preprint, Han \& Schied (2025) \cite{HS} used \follmer\!'s calculus to average over constant rebalanced portfolios.

However, in these existing continuous-time, model-free frameworks, the convergence results remain restricted or unaddressed. For the Markovian case, convergence was only established in the first-order sense i.e., \(t^{-1}\ln(W_t^*/\widehat W_t)\to0\), leaving the exact rate of convergence as an open question \cite[\S 1.1.3]{CSW}. Furthermore, this first-order convergence does not extend to more general allocations using the rough path approach; the asymptotic error bound in log-wealths is $O(t)$ for the Markovian case \cite[\S 4.12--4.14]{All}. This stands in sharp contrast to Cover's classical $O(\ln t)$ bound, which guarantees the rate of convergence \(t^{-1}\ln(W_t^*/\widehat W_t)=O(\ln t)/t\to0\). Furthermore, while \cite{HS} provides an elegant formulation for constant rebalanced portfolios, it leaves the tracking error unaddressed. 

Our results provide the first extension of Cover's theorem to continuous-time without probability. By successfully recovering the $O(\ln t)$ bound, we address an open question in the literature \cite[\S 1.1.3]{CSW} and resolve the current $O(t)$ limitations. Our results apply further to path-dependent allocations and \cadlag price paths. For a detailed comparison, we refer to the following Table 1.

\begin{table}[htbp]
\centering
\small
\setlength{\tabcolsep}{5pt}
\begin{tabular}{lcccccc}
\toprule
\textbf{Article} & \textbf{CT} & \textbf{Model-free} & \textbf{Allocation} & \textbf{\cadlag{}} & \textbf{Convergence} & \textbf{Error} \\
\midrule
Cover (1991) & $\times$ & $\checkmark$ & Constants (CRP) & $\times$ & $\checkmark$ & $O(\ln t)$ \\
Jamshidian (1992) & $\checkmark$ & $\times$ & Constants (CRP) & $\times$ & $\checkmark$ & $O(\ln t)$ \\
Cuchiero et al. (2019) & $\checkmark$ & $\checkmark$ & Markovian & $\times$ & $\checkmark$ & $O(t)$ \\
Allan et al. (2023) & $\checkmark$ & $\checkmark$ & Controlled path & $\times$ & Markovian & $O(t)$ \\
Han \& Schied (2025) & $\checkmark$ & $\checkmark$ & Constants (CRP) & $\times$ & $\times$ & $\times$ \\
\midrule
This Article & $\checkmark$ & $\checkmark$ & Path-dependent & $\checkmark$ & $\checkmark$ & $O(\ln t)$ \\
\bottomrule
\end{tabular}
\caption{Comparison of continuous-time universal portfolio results.}
\label{tab:universal_portfolio_comparison}
\end{table}

\section{Notations}\label{sec:not}\noindent

	Denote $D$ to be the Skorokhod space of $\R^{m}$-valued positive \cadlag functions \begin{eqnarray*}t\longmapsto x(t):=(x_1(t),\ldots,x_{m}(t))'\end{eqnarray*} on $\R_{+}:=[0,\infty)$ and for $p\in 2\mathbb{N}$, we denote $D(\R_{+},\R^{m}\otimes^p)$ the Skorokhod space of $\R^{m}\otimes^p$-valued \cadlag functions on $\R_{+}:=[0,\infty)$. Denote $C$, $\SS$, $BV$ respectively, the subsets of continuous functions, step functions, locally bounded variation functions in $D$. $x(0-):=x_0>0$ and $\Delta x(t):=x(t)-x(t-)$. The path $x\in D$ stopped at $(t,x(t))$ (resp. $(t,x(t-))$)\begin{eqnarray*}s\longmapsto x(s\wedge t)
\end{eqnarray*}shall be denoted by $x_t\in D$ (resp. $x_{t-}:=x_t-\Delta x(t)\I_{[t,\infty)}\in D$). We write $(D,\d)$ when $D$ is equipped with a complete metric $\d$ which induces the Skorokhod (a.k.a. \textsc{J}$_1$) topology. 

		Let $\pi:=(\pi_n)_{n\geq 1}$ be a fixed sequence of partitions $\pi_n=(t^{n}_0,...,t^{n}_{k_n})$ of $[0,\infty)$ into intervals $0=t^{n}_0<...<t^{n}_{k_n}<\infty$; $t^{n}_{k_n}\uparrow\infty$ with vanishing mesh $|\pi_n| \downarrow 0$ on compacts. By convention, $\max(\emptyset\cap\pi_n):=0$, $\min(\emptyset\cap\pi_n):=t^{n}_{k_n}$. Since $\pi$ is fixed, we will avoid superscripting $\pi$. 
		
   For any $p\in 2\mathbb{N}$, we say that $x\in D$ has finite $p$-th order variation $[x]_{p}$ if \begin{eqnarray*}\label{eq:qv}
   \sum_{\pi_n\ni t_i\leq t}\left(x({t_{i+1}})-x({t_{i}})\right)^{\otimes p}
    \end{eqnarray*}converges to $[x]_{p}$ in the Skorokhod J$_1$ topology in $D(\R_{+},\R^{m}\otimes^p)$, where $\otimes^p$ denotes the $p$-fold tensor product.  In light of \cite{CC}, we remark that in the special case $p=2$, this definition is equivalent to that of \follmer \cite{HF}. We denote $V_{p}$ the set of \cadlag paths of finite $p$-th order variations and write $QV:=V_2$, \begin{eqnarray*}t'_n:=\max\{t_i<t|t_i\in\pi_n\},\label{eq:t_n}\end{eqnarray*}and the following piecewise constant approximations of $x$ by

\begin{eqnarray}
x^{n}:=\sum_{t_i\in\pi_n}x(t_{i+1})\I_{[t_i,t_{i+1})}.\label{eq:x_n}
\end{eqnarray}We let $\Omega\subset D$ be \emph{\closed}{} (\cite[Def. 3.1]{CC3}) and define our \emph{domain} as\begin{eqnarray*}\Lambda:=\{(t,x_t)|t\in\R_{+}, x\in \Omega\}.\end{eqnarray*}We denote $\Lambda_{+}:=\{(t,x_t)|t\in(0,\infty), x\in \Omega\}.$ For a functional $F$ on $\Lambda$, we write $F_0:=F(0,x_0)$. For real-valued matrices of equal dimension, we write $\langle\cdot,\cdot\rangle$ to denote the Frobenius inner product and $|\cdot|$ to denote the Frobenius norm. If $f$ (resp. $g$) are  $\R^{m\times m}$-valued functions on $[0,\infty)$, we write\begin{eqnarray*}
		\int_{0}^{t} fdg:=\sum_{i,j}\int_{0}^{t}f_{i,j}(s-)dg_{i,j}(s)
		\end{eqnarray*}whenever the RHS makes sense. For two vectors $v$, $w>0$ of equal dimension $m$, we write \begin{eqnarray*}\frac{v}{w}:=(v_1/w_1,\ldots,v_m/w_m)'\end{eqnarray*} for the component-wise division.

\section{Foundation}

To ensure this article is self-contained, we recite key definitions and results from causal functional calculus and model-free finance, as presented in \cite[\S 3, \S 4]{CC3}, since they form the foundation for the subsequent sections.

The new additions to this article are: We present in Example \ref{Eg:counter} of an integrand whose left-Riemann sums exist (i.e., its path integral may be defined) but does not constitute a self-financing trading strategy. We introduce  relative arbitrage in Def. \ref{def:RA} and show that it does not exist on any generic domain.

\subsection{Causal functional calculus}\label{sec:calculus}\noindent 
    
\begin{Def}[Generic scenarios]\label{def:closed}
A non-empty subset $\Omega\subset D$ is called \emph{\closed} if $\Omega$ satisfies the following  closure properties under operations: (we recall (\ref{eq:x_n}) for the definition of $x^n$) 
	
	\begin{itemize}
	
	\item[(i)] For every $x\in \Omega$, $T> 0$, $\exists N(T)\in\N$;  $x_{T}^n \in \Omega ,\quad\forall n\geq N(T)$.
	
	\item[(ii)] For every $x\in\Omega, t\geq 0$, $\exists$ convex neighbourhood $\Delta x(t)\in \U$ of $0$; \begin{eqnarray}\label{eq:nbr}x_{t-}+e\I_{[t,\infty)}\in \Omega,\quad \forall e\in \U.\end{eqnarray}

    \end{itemize}
	\end{Def}

\begin{Rem}[Generic scenarios]\label{rem:closed}
    The motivation underlying generic scenarios is rooted in model uncertainty and robustness, ensuring that a strategy's PnL remains well-defined across an envelope of market deviations. 
    
    If a scenario $x \in \Omega$ is considered, its operational manifestations---namely, fine piecewise constant approximations (tick observations) and infinitesimal vertical perturbations (microstructure noise)---must be accommodated. This ensures the PnL remains calculable under these path variations rather than just a single speculated scenario.
\end{Rem}
	
	\begin{Eg}\label{eg:closed}Examples of \closed\  subsets include
	$\SS$, $BV$, $D$ and $V_{p}$  for $p\in 2\mathbb{N}$. Generic subsets are closed under finite intersections.  All subsets of $C$ are not \closed. For $-1<a<0<b$, \[\Omega:=\left\{x\in D \left\vert  a<\frac{\Delta x(t)}{x(t-)}< b\right., \forall t \right\}
\]is generic.
	\end{Eg}

    \begin{proof}
    To prove the last statement, suppose (i) does not hold given $x\in \Omega$, $T> 0$, then one may extract an increasing sub-sequence $(n_k)_{k}$ and a sequence of times $t_k\in \pi_{n_k}$; $t_k\rightarrow t\in[0,T]$ and wlog:  \[
\frac{\Delta x^{n_k}(t_{k})}{x^{n_k}(t_{k}-)}=\frac{x(t_{k+1})-x(t_{k})}{x(t_{k})}\geq b.
\]Since $x$ is \cadlag and $|\pi_{n_{k}}|\downarrow 0$ on $[0,T]$, it follows $\lim_k\frac{\Delta x^{n_k}(t_{k})}{x^{n_k}(t_{k}-)}=\frac{\Delta x(t)}{x(t-)}<b$, i.e. we have arrived at a contradiction. For (ii), we put $\U:=\{e|x_i(t-)a<e_i<x_i(t-)b\}$.
    \end{proof}

\begin{Def}[Strictly causal functionals]\label{def:causal} Let $F: \Lambda \to \mathbb{R}$   and denote $F_{-}(t,x_t)= F(t,x_{t-})$. $F$ is called \emph{strictly causal} if $F=F_{-}$.
\end{Def}

	We associate with the sequence of partitions $\pi$ a topology on the space $\Lambda$ of \cadlag paths called the $\pi$-topology, introduced in \cite{CC2}:
\begin{Def}[Continuous functionals]\label{prop:pi}
We denote by $C(\Lambda)$ the set of maps $F:\Lambda\to \mathbb{R}$ which satisfy 
\begin{alignat*}{3}
1.&(a) \lim_{s\uparrow t; s\leq t}F(s,x_{s-})=F(t,x_{t-}),\\ 
  &(b) \lim_{s\uparrow t; s<t}F(s,x_{s})=F(t,x_{t-}),\\
  &(c)\hspace{1mm} t_n\longrightarrow t; t_n\leq t'_n \Longrightarrow F(t_n,x^{n}_{t_n-})\longrightarrow F(t,x_{t-}),\\
  &(d)\hspace{1mm} t_n\longrightarrow t; t_n<t'_n \Longrightarrow F(t_n,x^{n}_{t_n})\longrightarrow F(t,x_{t-}),\\
\newline\\
2.&(a) \lim_{s\downarrow t; s\geq t}F(s,x_{s})=F(t,x_{t}),\\
  &(b) \lim_{s\downarrow t; s>t}F(s,x_{s-})=F(t,x_{t}),\\
  &(c)\hspace{1mm} t_n\longrightarrow t; t_n\geq t'_n\Longrightarrow F(t_n,x^{n}_{t_n})\longrightarrow F(t,x_t),\\
  &(d)\hspace{1mm} t_n\longrightarrow t; t_n>t'_n\Longrightarrow F(t_n,x^{n}_{t_n-})\longrightarrow F(t,x_{t}),\\
\end{alignat*}for all $(t,x_t)\in\Lambda$. 
A functional is called \emph{left (resp. right) continuous} if it satisfies 1.(a)-(d) (resp. 2.(a)-(d)). \end{Def} 

\begin{Def}[Regulated functionals]\label{def:version} A functional $F: \Lambda \to \mathbb{R}$ is \emph{regulated} if there exists $\widetilde{F}\in C(\Lambda)$ such that $\widetilde{F}_{-}=F_{-}$. The \emph{continuous version} $\widetilde{F}$ is unique by Def.~\ref{prop:pi}.2(b).
\end{Def}

\begin{Eg}
Let $f$ be a continuous function, then $F(t,x_t):=f(x(t-))$ is regulated and the continuous version is $\widetilde{F}=f(x(t))$.
\end{Eg}

\begin{Rem}\label{rem:regulated}
Since $C(\Lambda)$ is an algebra, we remark the set of regulated functionals forms an algebra.
\end{Rem}

\begin{Def}[Horizontal differentiability]\label{def:dt}
$F:\Lambda\longmapsto \R$ is called \emph{differentiable in time} if \begin{eqnarray*}
\D F(t,x_t):=\lim_{h\downarrow 0}\frac{F(t+h,x_t)-F(t,x_t)}{h} 
\end{eqnarray*} exists $\forall (t,x_t)\in\Lambda$.\end{Def}

\begin{Def}[Vertical differentiability]\label{def:dx}
$F:\Lambda\longmapsto \R$ is called \emph{vertically differentiable} if for every $(t,x_t)\in\Lambda$, the map 
\begin{eqnarray*}e\longmapsto F\left(t,x_t+e\I_{[t,\infty)}\right)\end{eqnarray*} is differentiable at $0$. We define $\nabla_{x}F(t,x_t):=(\nabla_{x_1}F(t,x_t),\ldots,\nabla_{x_m}F(t,x_t))'$;
\begin{eqnarray*}
\nabla_{x_i}F(t,x_t):=\lim_{\epsilon\rightarrow 0}\frac{F\left(t,x_t+\epsilon\mathbf{e}_i\I_{[t,\infty)}\right)-F(t,x_t)}{\epsilon}.
\end{eqnarray*}
\end{Def}

\begin{Def}[Differentiable]
A functional is called \emph{differentiable} if it is horizontally and  vertically differentiable.
\end{Def}

\begin{Rem}
All definitions above are extended to multidimensional functions on $\Lambda$ whose components satisfy the respective conditions.
\end{Rem}

\begin{Lem}\label{Lem:causal}
A function on $\Lambda$ is strictly causal if and only if it is differentiable in space with vanishing derivative.
\end{Lem}

\begin{proof}
We refer to \cite[Prop. 4.3]{CC2}.
\end{proof}

\begin{Def}[Classes $\X$ and $\M$]\label{def:smooth}
A continuous and differentiable functional $F$ is of \emph{class $\X$} if $\D F$ is right continuous and locally bounded, $\nabla_{x}F$ is left continuous and strictly causal. If in addition, $\D F$ vanishes, then $F$ is of \emph{class $\M$}. 
\end{Def}

We denote $\M(\Lambda)$ the set class $\M$ functionals, $\M_{0}(\Lambda):=\{M_0=0|M\in\M\}$ and $\M_{+}(\Lambda):=\{M>0|M\in\M\}$. 

\begin{Def}[Pathwise integral]\label{def:pathwise}Let $\phi:\Lambda\longmapsto \R^{m}$; $\phi_{-}$ be left continuous. For every $x\in\Omega$, define \begin{eqnarray}\label{eq:discrete}\II(t,x^n_{t}):=\sum_{\pi_n\ni t_i\leq t}\phi(t_i,x^n_{t_i-})\cdot(x(t_{i+1})-x(t_i)).\label{eq:lrs}\end{eqnarray} If $\II(t,x_t):=\lim_{n} \II(t,x^n_{t})$ exists and $\II$ is continuous, then $\phi$ is called \emph{integrable} and $\II:=\int\phi dx$ is called the \emph{pathwise integral}.
\end{Def}

\begin{Thm}\label{thm:ftc} 
A functional $F:\Lambda \to \mathbb{R}$ is a pathwise integral if and only if  $F\in \M_{0}(\Lambda)$\end{Thm}

\begin{proof}
We refer to \cite[Prop. 5.7, Thm. 5.8 \& Cor. 5.10]{CC2}.\end{proof}

\subsection{Trading strategy, self-financing and arbitrage}\label{sec:mf}\noindent 
    We consider a frictionless market with $d>0$ tradable assets, and one numeraire whose price is identically 1. We denote $x$ to be the price paths of tradable assets and $x\in\Omega$, where $\Omega$ is generic Def.~\ref{def:closed}. The \emph{number of shares} in assets $\phi$ and the numeraire $\psi$ held immediately before the portfolio revision at time $t$ will be denoted by $\phi_{-}$ and $\psi_{-}$. A trading strategy, aka portfolio, is a pair $(\phi,\psi)$ of regulated functionals $\phi: \Lambda \mapsto \R^{d}$ and $\psi: \Lambda \mapsto \R$. The value $V$ of the portfolio is given by \begin{eqnarray}
V(t,x_t):=\widetilde{\phi}(t,x_{t})\cdot x(t)+\widetilde{\psi}(t,x_{t}).\label{eq:self1}
\end{eqnarray}

  A key concept in mathematical finance is self-financing. This concept is usually defined by the choice of a path integration concept. This approach may be intuitive but dangerously imprecise as pointed out by Lyon \cite[\S 2.2]{TL} (see Example~\ref{Eg:counter}). The notion of self-financing thus hinges upon the choice of a particular integration concept. The following concept is based on {\it local} properties, without requiring the choice of a particular integration concept.
  
\begin{Def}[Self-financing trading strategy]\label{def:self}\noindent\\
A trading strategy aka portfolio $(\phi,\psi)$ is called \emph{self-financing} if for every $(t,x)\in \Lambda$;
\begin{itemize}
    \item [(i)] $\Delta \widetilde{\phi}(t,x_t)\cdot x(t)+\Delta\widetilde{\psi}(t,x_t)=0$,
    \item [(ii)] $\left(\widetilde{\phi}(t+h,x_t)-\widetilde{\phi}(t,x_t)\right)\cdot x(t)+\widetilde{\psi}(t+h,x_t)-\widetilde{\psi}(t,x_t)=0$ whenever $h\geq 0$. 
\end{itemize}
\end{Def}
    Both conditions correspond to the property that the proceeds from any change in the assets' position is financed by a corresponding change in the cash position. 

\begin{Rem}
If $(\phi,\psi)$ is self-financing, then the value of the portfolio may also be expressed as\begin{eqnarray}
V(t,x_t)=\phi(t,x_{t-})\cdot x(t)+\psi(t,x_{t-}).\label{eq:self2}
\end{eqnarray}
\end{Rem}

\Eg[A counter-example]\label{Eg:counter}\noindent\\
Let $\bigcup_{n} \pi_n\subset Q$. Define $\phi(t,x_t):=\I_{\{Q\}}(t)$, then $\phi$ is \textbf{not} self-financing because it is not regulated (i.e. left/right limits do not exist, instantaneous change of positions is indeterminate, so cannot be computed and implemented).

Nevertheless, the left-Riemann sums (\ref{eq:lrs}) converge to $x(t)-x(0)$. Note that $\phi$ is integrable in \ito's sense because $\phi$ is deterministic and bounded. (For a Brownian motion, the value of this integral is null).

\begin{Thm}[Representation]\label{thm:self}
Let $V$ be the value of the portfolio $(\phi,\psi)$. Then  $(\phi,\psi)$ is self-financing if and only if 
   $V\in\M(\Lambda)$; $\nabla_{x}V=\phi_{-}$ i.e.  \begin{eqnarray}
     V(t,x_t)=V(0,x_0)+\int_{0}^{t}\phi(s,x_{s-})dx.\label{eq:gain}
    \end{eqnarray}
\end{Thm}

\begin{proof}
We refer to \cite[Thm. 4.3]{CC3}.
\end{proof}

\begin{Prop}[Equivalence]\label{prop:self}
Let $V$ be a functional, the following are equivalent:
\begin{itemize}
    \item [(i)] $V$ is the value of a self-financing portfolio $(\phi,\psi)$.
    \item [(ii)] $V\in\M(\Lambda)$ and $\nabla_{x}V$ is regulated.
\end{itemize}
\end{Prop}

\begin{proof}
We refer to \cite[Prop. 4.4]{CC3}.
\end{proof}

\begin{Rem}[Self-financing $V$]\label{rem:self}
In view of Thm.~\ref{thm:self}, we may call a functional $V$ \emph{self-financing} if $V\in\M$ with regulated $\nabla_{x}V$. In particular, the self-financing trading strategy associated with $V$ is given by $\phi:=\nabla_{x}V$ and   \begin{eqnarray}\psi(t,x_t):=V(t,x_t)-\widetilde{\phi}(t,x_t)\cdot x(t).\label{eq:self3}\end{eqnarray}
\end{Rem}

\begin{Def}[Arbitrage]
Let $V$ be self-financing. We say that $V$ is an arbitrage if there exists $T>0$; $V(T,x_T)-V(0,x_0)\geq 0$ and there exists $x\in\Omega$; $V(T,x_T)-V(0,x_0)>0$.
\end{Def}

\begin{Thm}\label{thm:NA}
Arbitrage does not exist in a generic market.
\end{Thm}

\begin{proof}
We refer to \cite[Thm. 4.8]{CC3}.  
\end{proof}

\begin{Def}[Relative arbitrage]\label{def:RA}
Let $V, W$ be self-financing and $V_0=W_0$. We say that $V$ is an arbitrage relative to $W$ if there exists $T>0$; $V(T,x_T)-W(T,x_T)\geq 0$ and there exists $x\in\Omega$; $V(T,x_T)-W(T,x_T)>0$.
\end{Def}

\begin{Cor}\label{Cor:NA}
Relative arbitrage does not exist in a generic market.
\end{Cor}

\begin{proof}
It is an immediate consequence of Thm.\ref{thm:NA} and that $V-W\in \M_0$.
\end{proof}

\begin{Rem}[Relative arbitrage]
Generic domains include, but are not limited to, the set of simple step functions $\SS$. In contrast, a subset of continuous paths---the domain typically used in stochastic portfolio theory---is not generic. Hence, there is no contradiction to the existence of relative arbitrage results established on a subset of continuous paths, such as in \cite[Karatzas \& Kim (2020)]{KK}. See also Rem.~\ref{rem:closed}.
\end{Rem}

\section{Portfolio allocation}\label{sec:mf2}\noindent

In this section, we extend the concept of self-financing to a portfolio allocation strategy (a functional taking values in a simplex) and associate this strategy with a path-dependent partial differential equation (PPDE). We prove that a portfolio allocation strategy is self-financing if and only if its associated PPDE admits a solution, which characterizes the evolution of wealth associated with the strategy. We derive the general explicit solution for any generic domain and provide examples of concrete solutions. 

\subsection{Self-financing allocation strategy}
    For the study of investment problems, we shall be focusing on the \emph{allocation} strategy $\theta:=(\theta_{1},\ldots,\theta_{d})'$ whose individual components are fraction of the portfolio value in the respective asset. 

    \begin{Def}[Allocation strategy]\label{def:alloc}An \emph{allocation strategy} or \emph{allocation} $\theta$ is a $\R_{+}^{d}$-valued regulated functional on $\Lambda$ satisfying $\sum_{i=1}^{d}\theta_i\leq 1$. 
\end{Def}
    
    In particular, if $V$ is the value of a portfolio $(\phi,\psi)$ that implements an allocation $\theta$, then the identities \begin{alignat*}{2}
\tilde{\phi}_{i}x_i=\widetilde{\theta}_{i}V  \quad\forall i, \quad
    \tilde{\psi}=\left(1-\sum_i \widetilde{\theta}_{i}\right)V 
    \end{alignat*}hold at all times.

\begin{Def}[Self-financing allocation]\label{def:alloc_self}Let $\theta$ be an allocation strategy and $V>0$ be continuous. We associate the pair $(\theta$, $V)$ with the following trading strategy: \begin{alignat}{2}
     \phi(t,x_{t})&:=\left(\frac{\widetilde{\theta}}{x}V\right)(t,x_t):=\left(\frac{{\widetilde{\theta}}_1(t,x_{t})}{x_1(t)},\ldots,\frac{\widetilde{\theta}_d(t,x_{t})}{x_d(t)}\right)'V(t,x_t),\nonumber\\
     \psi(t,x_{t})&:=\left(1-\sum_{i=1}^{d}{\widetilde{\theta}_{i}(t,x_{t})}\right)V(t,x_t).\label{eq:implement}
    \end{alignat}The portfolio $(\phi,\psi)$ is called an \emph{implementation} of $\theta$. An allocation strategy is called \emph{self-financing} if there exists a \emph{self-financing implementation}. We denote $\Theta(\Lambda)$ the set of all self-financing allocation strategies on $\Lambda$.
\end{Def}   

\begin{Rem}
    One verifies the identity $\tilde{\phi}_{i}x_i=\widetilde{\theta}_{i}V$ for each $i$, and hence \begin{eqnarray}\widetilde{\phi}\cdot x+\widetilde{\psi}=\left(\sum_{i}\widetilde{\theta}_{i}   \right)V+\left(1-\sum_{i}\widetilde{\theta}_{i}\right)V=V,\label{eq:implement2}\end{eqnarray} i.e. $V$ is the value of the portfolio $(\phi,\psi)$ that implements $\theta$. Every allocation strategy has an implementation. An implementation is not necessarily self-financing! 
\end{Rem}

\begin{Thm}[Wealth equation]\label{thm:alloc_self}
An allocation strategy $\theta$ is self-financing if and only if there exists $V\in\M_{+}$; $V$ solves the following (path-dependent) partial differential equation on $\Lambda_{+}$:
\begin{eqnarray}
\nabla_{x}V-\left(\frac{\theta}{x}V\right)_{-}=0.\label{eq:pde}
\end{eqnarray}
\end{Thm}

\begin{proof}
If $V\in \M_{+}$ solves (\ref{eq:pde}), we define $(\phi,\psi)$ by (\ref{eq:implement}) and obtain $\nabla_{x}V=\phi_{-}$. It follows from (\ref{eq:implement2}), Thm.~\ref{thm:self} and Def.~\ref{def:alloc_self} that $\theta$ is self-financing. On the other hand, if $\theta$ is self-financing, then by Def.~\ref{def:alloc_self}, there exists a continuous $V>0$ such that the trading strategy (\ref{eq:implement}) is self-financing. By (\ref{eq:implement}), (\ref{eq:implement2}) and Thm.~\ref{thm:self}, we obtain $V\in \M_{+}$; $V$ solves (\ref{eq:pde}).\end{proof}

\begin{Rem}[Economic interpretation]
The PPDE \eqref{eq:pde} characterises the market-driven PnL. The wealth sensitivity to a price tick $\nabla_{x}V$ must exactly equal the shares held prior to the perturbation $\left(\frac{\theta}{x}V\right)_{-}$. Practically, this ensures the allocation strategy operates with zero exogenous cash flows; wealth fluctuations are strictly driven by existing positions reacting to market movements.
\end{Rem}

\begin{Prop}[Explicit solution]\label{prop:general}
If $V\in\M_{+}$ solves (\ref{eq:pde}), then the explicit solution is given by\begin{eqnarray}
V(t,x_t)=V_0 \lim_{n\rightarrow\infty}\prod_{\pi_n\ni t_i\leq t}\left(1+\theta(t_i,x^n_{t_i-})\cdot \frac{\Delta x^n(t_i)}{x(t_i)}\right).\label{eq:general}
\end{eqnarray}
\end{Prop}

\begin{proof}
By Thm.~\ref{thm:ftc}, we see that $V-V_0$ is a pathwise integral. By Def.~\ref{def:pathwise}, (\ref{eq:pde}) and the observation that $$
V(t_{i-1},x^{n}_{t_{i-1}})=V(t_{i},x^{n}_{t_i-})
,$$ we obtain \begin{alignat}{2}
V(t_{i},x^{n}_{t_{i}})&=V(t_{i},x^{n}_{t_{i-}})+V(t_i,x^n_{t_i-})\frac{\theta(t_i,x^n_{t_i-})}{x^n(t_i-)}\cdot\left(x(t_{i+1})-x(t_i)\right)\nonumber\\
&=V(t_{i},x^{n}_{t_{i-}})\left(1+\theta(t_i,x^n_{t_i-})\cdot \frac{\Delta x^n(t_i)}{x(t_i)}\right)\label{eq:diff_eq}
\end{alignat}for every $0<t_i\in\pi_n$. By the continuity of V, the proof is complete.
\end{proof}

\begin{Rem}[Path-dependency]
Even if the allocation $\theta$ were to be a constant (without a unit component), the solution (\ref{eq:general}) is in general \emph{path-dependent}, see also \cite[4.16]{CC2}.\end{Rem}

\begin{Cor}[Uniqueness]\label{cor:unique}
Let $\theta$ be an allocation, $\xi\in\Xi$ where \begin{eqnarray}
\Xi:=\{V_0|V\in \M_{+}\}\label{eq:initial}
\end{eqnarray}and $U, W \in \M_{+}$ be two solutions to the following Cauchy (initial value) problem on $\Lambda_{+}$:

\begin{eqnarray}
\begin{cases}
\nabla_{x}V-\left(\frac{\theta}{x}V\right)_{-}=0\\
V_0=\xi
\end{cases}\label{eq:cauchy},
\end{eqnarray}then $U\equiv W$ on $\Lambda$.
\end{Cor}

\begin{proof}
It is an immediate consequence of Prop.~\ref{prop:general}.
\end{proof}

\begin{Def}[Pathwise exponential]\label{def:exponent}
Let $\theta$ be an allocation and $\xi\in\Xi$ as defined in (\ref{eq:initial}). For every $x\in\Omega$, define  \begin{eqnarray}
\E(t,x^{n}_t):=\xi\prod_{\pi_n\ni t_i\leq t}\left(1+\theta(t_i,x^n_{t_i-})\cdot \frac{\Delta x^n(t_i)}{x(t_i)}\right),\label{eq:dd}
\end{eqnarray}for all $t\geq 0$. If $\E(t,x_t):=\lim_n \E(t,x^{n}_t)>0$ exists and continuous on $\Lambda$, then we write $\E:=\E(\xi,\theta)$. We shall call $\E(\xi,\theta)$ the \emph{pathwise 
exponential} of $\theta$ with initial value $\xi$.
\end{Def}

\begin{Prop}[A sufficent condition]
Let $\theta$ be an allocation, $\xi\in\Xi$ and $\E(t,x^{n}_t)$ as defined in (\ref{eq:dd}). If for every $x\in\Omega$, $T>0$, the collection of step functions $$(t\mapsto \E(t,x^{n}_t))_{n}$$ is a Cauchy sequence in $D([0,T],\R)$ with regard to a complete J$_1$ metric and its limit is positive, then the pathwise exponential $\E(\xi, \theta)$ exists.   
\end{Prop}

\begin{proof}
We can write (\ref{eq:dd}) as a discrete pathwise integral modulo an initial value. The claim now follows from \cite[Thm. 5.6]{CC2}.
\end{proof}

\begin{Thm}[Existence]\label{thm:exist}
Let $\theta$ be an allocation and $\xi\in\Xi$. The following are equivalent:
\begin{itemize}
    \item [(i)] $\lim_n \E(t,x^{n}_t)$ in (\ref{eq:dd}) exists, is positive and continuous on $\Lambda$.

    \item [(ii)] The pathwise exponential $\E(\xi,\theta)$ exists and is the unique $\M_{+}$ solution to the Cauchy problem (\ref{eq:cauchy}).

    \item [(iii)] $\theta$ is self-financing. 
\end{itemize}
\end{Thm}

\begin{proof}
If (i) holds, we put $z:=x+e\I_{[t,\infty)}\in \Omega$ (for sufficiently small $e$). For $t>0$, we observe \begin{alignat*}{2}
\E(t,z_t)-\E(t,x_t)&=\lim_{n}\left(\E(t,z^{n}_t)-\E(t,x^{n}_t)\right)\\
&=\lim_n \left(\frac{\theta(t'_n,x^n_{t'_n-})}{x(t'_n)}\E(t'_n, x^n_{t'_n-} )\right)\cdot e\\
&=\left(\frac{\theta(t,x_{t-})}{x(t-)}\E(t, x_{t-} )\right)\cdot e,
\end{alignat*}by the continuity of $\E$ and the left-continuity of $\theta_{-}$ and $x$. From the above equality, we see that $\nabla_x \E=(\frac{\theta}{x}\E)_{-}$ on $\Lambda_{+}$ i.e. $\E$ solves (\ref{eq:pde}). Therefore, we also see that $\nabla_x \E$ is strictly causal. Since $\mathcal E(t+h,x_t)=\mathcal E(t,x_t)$ for all $h>0$, it follows $\D \E=0$. By Def.~\ref{def:smooth}, Prop.~\ref{prop:general} and Cor.~\ref{cor:unique}, we have established that $\E\in \M_{+}$ and arrived at (ii). By Thm.~\ref{thm:alloc_self}, we proceed to (iii) and finally if $V\in M_{+}$ solves (\ref{eq:pde}), we put $U:=\frac{V}{V_0}\xi$, then $U\in \M_{+}$ also solves (\ref{eq:pde}) with $U_0=\xi$. By Thm.~\ref{thm:alloc_self},  Prop.~\ref{prop:general}, we deduce (i).
\end{proof}

\begin{Cor}\label{cor:ddmap}
The $\M_+$-solution map to the Cauchy problem (\ref{eq:cauchy})\begin{alignat}{2}
\E:\Xi\times\Theta(\Lambda)&\longmapsto \M_{+}(\Lambda)\nonumber\\
(\xi, \theta)&\longmapsto \E(\xi,\theta)\label{eq:ddmap}
\end{alignat}is well-defined.
\end{Cor}

\begin{proof}
It is an immediate consequence of Thm.~\ref{thm:exist}.
\end{proof}

\begin{Rem}\label{rem:ddprop}
$\alpha \E(1,\theta)=\E(\alpha,\theta)$, for all $\alpha>0$ due to (\ref{eq:dd}).
\end{Rem}

\subsection{Examples}

We now provide examples of explicit solution to the Cauchy problem which characterises the evolution of portfolio value (a.k.a. wealth) associated with a given self-financing allocation across generic scenarios. 

\begin{Eg}
If $\Omega\subset \SS$, then every allocation is self-financing. 
\end{Eg}

\begin{proof}
If $\Omega\subset\SS$, then (\ref{eq:dd}) becomes a finite product.
\end{proof}

\begin{Eg}[Single stock]\label{Eg:single}
Let $1\leq i \leq d$ and $\theta(t,x_t):=\mathbf{e}_i$. Then $\theta\in\Theta(\Lambda)$ and 
$$\E(1,\theta)(t,x_t)=\frac{x_i(t)}{x_i(0)}.$$
\end{Eg}

\begin{Eg}[Market index]\label{Eg:market}
Let $$\theta(t,x_t):=\frac{x(t)}{\sum_{i=1}^d x_i(t)}=\frac{1}{\sum_{i=1}^d x_i(t)}(x_1(t),\ldots,x_d(t))'.$$Then
$\theta\in\Theta(\Lambda)$ and \begin{eqnarray}
\E(1,\theta)(t,x_t)=\frac{\sum_{i=1}^d x_i(t)}{\sum_{i=1}^d x_i(0)}\leq \max_{1\leq i\leq d}\frac{x_i(t)}{x_i(0)}.\label{eq:ine}\end{eqnarray}
\end{Eg}


\begin{Eg}[Simple average]\label{Eg:avg}
Let $T>0$, $1\leq i \leq d$ and $$\theta(t,x_t):=\frac{ x_i(t)\left(1- \frac{t\wedge T}{T}\right) }{\frac{1}{T}\int_{0}^{t}{x_i(s)}ds+x_i(t)\left(1- \frac{t\wedge T}{T}\right) }\mathbf{e}_i.$$Then
$\theta\in\Theta(\Lambda)$ and $$
\E(1,\theta)(t,x_t)=\frac{1}{T}\int_{0}^{T}\frac{x_i(s\wedge t)}{x_i(0)}ds.$$
\end{Eg}

\begin{Eg}[Exponential average]\label{Eg:exp}
Let $\lambda>0$, $1\leq i \leq d$ and $$\theta(t,x_t):=\frac{ {x_i(t)} }{\lambda\int_{0}^{t}x_i(s) e^{\lambda (t-s) }ds+x_i(t)}\mathbf{e}_i.$$Then
$\theta\in\Theta(\Lambda)$ and $$
\E(1,\theta)(t,x_t)=\lambda \int_{0}^{\infty}\frac{x_i(s\wedge t)}{x_i(0)}e^{-\lambda s}ds.$$
\end{Eg}

\begin{Eg}[Portfolio of portfolio]\label{Eg:pop}
Let $T>0$, $\theta\in\Theta(\Lambda)$, $\E_{\theta}:=\E(1,\theta)$. Define a new portfolio
$$ 
\bar{\theta}(t,x_t):=\frac{\theta(t,x_t)\E_{\theta}(t,x_t)\left(1- \frac{t\wedge T}{T}\right)}{\frac{1}{T}\int_{0}^{t}\E_{\theta}(s,x_s)ds +\E_{\theta}(t,x_t)\left(1- \frac{t\wedge T}{T}\right)}.$$
Then
$\bar{\theta}\in\Theta(\Lambda)$ and $$
\E(1,\bar{\theta})(t,x_t)=\frac{1}{T}\int_{0}^{T}\E_{\theta}(s,x_{s\wedge t})ds.$$
\end{Eg}

\begin{proof}
The inequality (\ref{eq:ine}) follows from \cite[Lem. 1]{CO}. All examples follows from Thm.~\ref{thm:ftc}, Thm.~\ref{thm:alloc_self} and  Cor.~\ref{cor:ddmap} by simply differentiating $\E$ because $\D\E=0$ in all cases. For instance, we have for the market index: \[
\nabla_{x}\left(\frac{\sum_{i=1}^d x_i(t)}{\sum_{i=1}^d x_i(0)}\right)=\frac{1}{\sum_{i=1}^d x_i(0)}(1,\ldots,1)'=\left(\frac{\theta}{x}\left(\frac{\sum_{i=1}^d x_i(t)}{\sum_{i=1}^d x_i(0)}\right)\right)_{-} 
\] 
\end{proof}

\begin{Rem}
The above examples (except the first) hold on every generic domain. In the last examples, the allocation strategies are path-dependent and do not necessarily admit variation of any order. Further path-dependent examples are provided in the next sections, for instance E.g. \ref{cor:softmax}. Single asset strategies may be extended to multi-asset strategy by aggregating, which we demonstrate in the next section.
\end{Rem}


\section{Allocation algorithms}\label{sec:pa}\noindent

In this section, we extend two discrete-time machine learning algorithms to continuous-time meta-learning algorithms. These algorithms take multiple strategies as input and track the wealth generated by the best individual strategy and the best convex combination of strategies, respectively. We show that their tracking errors in log wealth are bounded by $O(1)$ and $O(\ln t)$, respectively.  

The first algorithm operates on any generic domain and is based on an aggregating algorithm of Vovk \cite{VV} (also known as Exponential weights or the Laissez-faire algorithm in the investment context \cite[\S 1.2]{YK}), which belongs to the class of online learning with expert advice algorithms. Using this algorithm, we demonstrate how to generate new strategies in closed form from known strategies.

The second algorithm operates on paths with finite quadratic variations and is based on Cover's original algorithm \cite{TC}, which only applies to constant rebalanced portfolios. We extend this algorithm to the convex hull generated by multiple strategies and prove Cover's main theorem \cite[Thm. 6.1]{TC} in this context. We recall $\E$, the solution map from the previous section (\ref{eq:ddmap}) and define \begin{alignat}{2}
W:\Theta(\Lambda)&\longmapsto \M_{+}(\Lambda)\nonumber\\
\theta&\longmapsto W(\theta):=\E(1,\theta)\label{eq:wmap}
\end{alignat}$W(\theta)$ is called the \emph{wealth} associated with the allocation strategy $\theta\in\Theta$. 

\subsection{Best individual strategy}

For $m\in \N$, we let $\theta^{(k)}\in\Theta$, $W_{k}:=W(\theta^{(k)})$ for $k=1,\ldots,m$ and denote $$B:=\left\{(b_1,\ldots,b_m)'| b_k>0; \sum b_k=1\right\}$$ the set of initial weights and $\bar{B}$ be its closure.

\begin{Thm}[Laissez-faire algorithm]\label{thm:AA}
For every $b\in \bar{B}$, we define \begin{eqnarray*}
\widehat{\theta}(b):=\frac{\sum_{k=1}^{m}\theta^{(k)}b_k W(\theta^{(k)})}{\sum_{k=1}^{m}b_k W(\theta^{(k)})}, 
\end{eqnarray*}then $\widehat{\theta}(b)\in\Theta$ and

\begin{eqnarray}
W(\widehat{\theta}(b))=\sum_{k=1}^{m}b_k W(\theta^{(k)}).\label{eq:laissez}\end{eqnarray}
\end{Thm}

\begin{proof}
Let us write $W_k:=W(\theta^{(k)})$ and $M:=\sum_{k}b_k W_k\in\M_{+}$ (convex cone), then $\nabla_x W_k=\left(\frac{\theta^{(k)}}{x}W_k\right)_{-}$ by (\ref{eq:pde}) and \begin{alignat*}{2}
\nabla_{x}M=\left(\frac{\sum_{k}\theta^{(k)}b_k W_k}{x}\right)_{-}=\left(\frac{\sum_{k}\theta^{(k)}b_k W_k/M}{x}M\right)_{-}
=\left(\frac{\widehat{\theta}}{x}M\right)_{-}\end{alignat*}by the linearity of the operators $\nabla_x$ and $(\cdot)_{-}$. It follows from Thm.~\ref{thm:alloc_self} that $\widehat{\theta}(b)\in\Theta$, the proof is complete by Thm.~\ref{thm:exist}(ii) and Cor.~\ref{cor:ddmap}. 
\end{proof}

\begin{Cor}[Bounds and asymptotic]\label{cor:AA}
Let $b\in B$ and \begin{alignat*}{2}
W^{*}:=\max_{k} W_k, \quad \widehat{W}(b)&:=W(\widehat{\theta}(b)).
\end{alignat*}Then 
\begin{eqnarray}
1\leq \frac{W^{*}}{\widehat{W}}\leq \max_k \frac{1}{b_k}.\label{eq:aa_asym}
\end{eqnarray}In particular, 
$$\frac{1}{T}\ln\left(\frac{W^{*}_T}{\widehat{W}_T}\right)\leq\frac{1}{T}\ln\left(\frac{1}{\min_k b_k}\right)\rightarrow 0,$$ as $T\uparrow\infty$.
\end{Cor}

\begin{proof}
It is an immediate consequence of Thm.~\ref{thm:AA}.
\end{proof}

\begin{Rem} 
The upper bound in (\ref{eq:aa_asym}) is minimised at $b_k\equiv 1/m$.
\end{Rem}

\begin{Cor}[Existence of optimum]\label{thm:minimax}
\begin{eqnarray}
1\leq \inf_{b\in \bar{B}}\sup_{\Lambda}\left(\frac{W^*}{\widehat{W}(b)}\right)\leq m,\label{eq:minmax}
\end{eqnarray}In particular, the infimum is attained.
\end{Cor}

\begin{proof}
By Cor.~\ref{cor:AA}, We have $1\leq \sup_{\Lambda}\left(\frac{W^*}{\widehat{W}(b)}\right)\leq \max_k\frac{1}{b_k}$ for every $b\in B$. Put $b_k:\equiv1/m$, the lower and upper bounds are thus established. By the definition of infimum and the compactness of $\bar B$, there exists a sequence $b_{n}\rightarrow b^{*}\in\bar B$; 
$$
\left(\frac{W^*}{\widehat{W}(b_n)}\right)\leq \sup_{\Lambda}\left(\frac{W^*}{\widehat{W}(b_n)}\right)\rightarrow\inf_{b\in \bar B}\sup_{\Lambda}\left(\frac{W^*}{\widehat{W}(b)}\right).
$$Since the map $b\mapsto\widehat{W}(b)$ is continuous on $\bar B$ at every point of $\Lambda$, we obtain $$
\left(\frac{W^*}{\widehat{W}(b^{*})}\right)\leq\inf_{b\in \bar B}\sup_{\Lambda}\left(\frac{W^*}{\widehat{W}(b)}\right)
$$at every point of $\Lambda$, hence the infimum is attained at $b^{*}\in\bar{B}$.
\end{proof}

\begin{Rem} 
The minmax problem (the least upper bound and optimal $b^{*}$) in \ref{eq:minmax} may be solved for by generating scenarios for $\Omega$. The least upper bound, which may be strictly less than $m$, will be valid if the realised scenario is among the generated ones.     
\end{Rem}

\begin{Eg}[Best stock]\noindent\\
Let $\theta^{(i)}(t,x_t):=\mathbf{e}_i$ for $i=1,\ldots,d$. Then
$\theta^{(i)}\in\Theta(\Lambda)$ for $i=1,\ldots,d$, $$W^{*}(t,x_t)=\max_{i} \frac{x_i(t)}{x_i(0)},$$ 
$$
\widehat{\theta}(t,x_t)=\frac{1}{\sum_{i=1}^{d}b_i\frac{x_i(t)}{x_i(0)}}\left(b_1\frac{x_1(t)}{x_1(0)},\ldots,b_d\frac{x_d(t)}{x_d(0)}   \right)'\in\Theta(\Lambda)
$$ and
$$
\widehat{W}(t,x_t)=\sum_{i=1}^{d}b_i\frac{x_i(t)}{x_i(0)}.
$$
\end{Eg}

\begin{proof}
It is an immediate consequence of Eg.~\ref{Eg:single} and Thm.~\ref{thm:AA}.
\end{proof}

\begin{Eg}[Best of Final Wealth and Time Average]\label{eg:insured}\noindent\\
Let $T>0$, $m:=2$, $\theta\in\Theta$ and $W:=W(\theta)$. Define\begin{alignat*}{2}\theta^{(1)}(t,x_t)&:=\theta\I_{[0,T]},\\
\theta^{(2)}(t,x_t)&:=\frac{ \theta W(t)\left(1- \frac{t\wedge T}{T}\right) }{\frac{1}{T}\int_{0}^{t}{W(s)}ds+W(t)\left(1- \frac{t\wedge T}{T}\right) }.
\end{alignat*}Then
$\theta^{(1)}, \theta^{(2)} \in\Theta(\Lambda)$, $$W^{*}(t,x_t)=W(t\wedge T)\bigvee \frac{1}{T}\int_{0}^{T}W(s\wedge t)ds,$$ 
\begin{eqnarray*}
\widehat{\theta}(t,x_t)=\frac{\theta W(t)\left(1-b_2\frac{t\wedge T}{T}\right)}{b_2\frac{1}{T}\int_{0}^{t\wedge T}W(s)ds +W(t)\left(1-b_2\frac{t\wedge T}{T}\right)}
\end{eqnarray*} and
\begin{eqnarray}
\widehat{W}(t,x_t)=(1-b_2)W(t\wedge T)+b_2\frac{1}{T}\int_{0}^{T}W(s\wedge t)ds.
\label{eq:insured}\end{eqnarray}
\end{Eg}

\begin{proof}
It is an immediate consequence of Eg.~\ref{Eg:pop} and Thm.~\ref{thm:AA}.
\end{proof}

\begin{Rem}
The above examples hold on every generic $\Omega$. In example \ref{eg:insured}, both the allocation strategy and its associated wealth (\ref{eq:insured}) are path-dependent and do not necessarily admit variation of any order. 
\end{Rem}

\subsection{Best convex combination of strategies}\label{sub:up}

In this section, we shall extend Cover's universal algorithm \cite{TC} to the convex hull \begin{eqnarray*}
\B:=\left\{\sum_{k=0}^{m}b_{k}\theta^{(k)} \Bigg{|} b_{k}\geq 0; \sum_{k=0}^{m}b_k=1\right\},
\end{eqnarray*}generated by finitely many strategies $\theta^{(0)}:=\mathbf{0}$ (i.e. pure cash), $\theta^{(k)}\in\Theta(\Lambda), k=1,\ldots, m$ under the domain $\Lambda$, where \begin{eqnarray}
\Omega:=\left\{x\in QV\Bigg| -\delta_{-}<\frac{\Delta x(t)}{x(t-)}<\delta_{+},\quad\forall t>0\right\}\label{eq:omega},
\end{eqnarray}$\delta_{-}\in(0,1)$ and $\delta_{+}>0$. We observe that $\Omega$ is generic (see Eg.~\ref{eg:closed}). Since $\theta^{(0)}$, $\theta^{(k)}$ for $k=1,\ldots m$ are fixed, we shall denote \begin{eqnarray*}\Delta_{m}:=\left\{ b:=(b_1,\ldots ,b_m)'\in\R_{+}^{m} \Bigg{|} \sum_{k=1}^{m} b_k\leq 1 \right\}\end{eqnarray*} to be an $m$-simplex and $\mathring{\Delta}_{m}$ its interior. For $b\in\Delta_{m}$, we shall write\begin{alignat}{2}W(b):=W(\theta(b)),\quad\theta(b):=\sum_{k=1}^{m}b_{k}\theta^{(k)},\label{eq:tb}\end{alignat}if $\theta(b)\in\Theta(\Lambda)$. We remark that $|\Delta_m|=\frac{1}{m!}$ and that $b\mapsto\theta(b)$ is a surjection from $\Delta_m$ onto $\B$. For $x\in QV$, we write \begin{eqnarray}
\Sigma(T,x_T):=\left(\int_0^{T}\frac{\theta^{(k)}}{x}\left(\frac{\theta^{(l)}}{x}\right)'d[x]\right)_{1\leq k,l\leq m}\label{eq:definite}
\end{eqnarray}and denote $\lambda_{\min}(t,x_t)$ 
the minimal (resp. $\lambda_{\max}(t,x_t)$ the maximal) eigenvalues of $\Sigma(t,x_t)$. 

\begin{Lem}\label{lem:definite}
$\Sigma(T,x_T)$ is positive semi-definite.
\end{Lem}

\begin{proof}
Since each $$b\longmapsto\ln\left(1+\theta(b)(t_i,x^{n}_{t_{i-}})\cdot\frac{\Delta x^n(t_{i})}{x(t_{i})}\right)$$ is a twice continuously differentiable concave function on $\Delta_m$, it follows
\begin{eqnarray*}\left(\left\langle \theta_{-}^{(k)}\theta_{-}^{(l)}{'},\frac{\Delta {x}^{n}(t_i)}{x(t_i)}\frac{\Delta {x}^{n}(t_i){'}}{x(t_i)}\right\rangle\right)_{1\leq k,l\leq m}\end{eqnarray*} is positive semi-definite and so is the finite sum $\Sigma(T,x^n_T)$. Since $\Sigma(T,x^n_T)\xrightarrow{n} \Sigma(T,x_T)$ \cite[Lem. 4.15]{CC2}. The proof is complete.
\end{proof}

\begin{Rem}\label{rem:definite}
$\Sigma(t,x_t)-\Sigma(s,x_s)$ is positive semi-definite for all $t>s$.  (i.e. $t\mapsto \Sigma(t,x_t)$ is monotonic increasing in the sense of Loewner order). $t\mapsto\lambda_{\min}(t,x_t)$ (resp. $\lambda_{\max}(t,x_t)$) are monotonic increasing. 
\end{Rem}

\begin{Eg}
Let $x\in\Omega$ and $\theta^{(k)}=e_k$ for $k=1,\ldots m=d$ (standard basis) and $$t\mapsto \left(\int_{0}^{t}\frac{1}{x_1(s)}dx_1(s),\ldots,\int_{0}^{t}\frac{1}{x_d(s)}dx_d(s)\right)'$$ be a sample path of a Brownian Motion with covariance matrix $\Sigma$. Denote $\lambda_{\max}$ for the maximal eigenvalue of $\Sigma$,  then $$
\Sigma(T,x_T)=\Sigma T
$$ and $\lambda_{\max}(T,x_T)=\lambda_{\max} T$.
\end{Eg}

\begin{proof}
It is a straightforward application of \follmer\!'s calculus \cite[Thm.\&(14)]{HF}.
\end{proof}

The outline for the rest of this section is as follow: we shall introduce the following lemmas to first establish in Thm.~\ref{thm:dct} that universal portfolio over the convex hull is self-financing and its wealth is the weighed average of wealths over the convex hull.  We then prove the main theorem Thm.~\ref{Thm:universal}, which gives an exact formula (instead of inequality) to express the wealth of the universal portfolio relative to the maximal wealth. Finally in Cor.~\ref{cor:asym}, we obtain the extension of Cover's theorem and conclude with a discussion of related literature in \ref{sec:RL}.

For the ease of exposition, we may use the following notations unless otherwise specified. If $(t,x_t)\in\Lambda$ is fixed but not the focal, we may suppress $(t,x_t)$ and write, for example:\begin{alignat*}{2}
W(b)&:=(W(b))(t,x_t),\\
W^n(b)&:=(W(b))(t,x^n_t).
\end{alignat*}Similarly, if $x\in \Omega$ is fixed but not the focal, we write:
$
W_t(b):=(W(b))(t,x_t)
$. We shall only use these notations for a functional on $\Lambda$.

\begin{Lem}[\ito's lemma]\label{lem:ito}
Let $\theta\in\Theta$, then
\begin{alignat}{2}
\ln W_{t}(\theta)&=\int_{0}^{t}\frac{\theta}{x}dx-\frac{1}{2}\int_0^{t}\frac{\theta}{x}\left(\frac{\theta}{x}\right)'d[x]\nonumber\\
&+\sum_{0<s\leq t}\ln\left(1+\theta_{-}\cdot \frac{\Delta x(s)}{x(s-)}\right)-\theta_{-}\cdot\frac{\Delta x(s)} {x(s-)}+\frac{1}{2}\left\langle\theta_{-}\theta_{-}',\frac{\Delta{x(s)}}{x(s-)}\frac{\Delta{x(s)}}{x(s-)}'\right\rangle,\label{eq:ito}
\end{alignat}where the absolute value of the series is bounded by \begin{eqnarray}
\frac{1}{2(1-\delta_{-})^2}\int_0^{t}\frac{\theta}{x}\left(\frac{\theta}{x}\right)'d[x].\label{eq:ito_bound}
\end{eqnarray}
\end{Lem}

\begin{proof}
By Thm.~\ref{thm:exist} and (\ref{eq:ddmap}), we can apply the change of variable formula of \cite[Thm. 5.14]{CC2} to the functional $\ln(W)$ to obtain (\ref{eq:ito}).
Let $O(t,x_t)$ denotes the series of (\ref{eq:ito}) and $O_n:=O(t,x^n_t)$. Since $\ln W$ and all individual terms other than $O$ are continuouos, it follows $O$ is continuous and
$O(t,x_t)=\lim_n O_n$. Apply a second order Taylor expansion to each individual log term of $O_n$ i.e.
$$e\longmapsto\ln\left(1+\theta(t_i,x^{n}_{t_{i-}})\cdot\frac{e}{x(t_{i})}\right)$$
and observe that $\frac{\Delta x^n(t_i)}{x(t_i)}\frac{\Delta x^n(t_i)'}{x(t_i)}$ is positive semi-definite, we see that $|O_n|$ is bounded by (\ref{eq:ito_bound}) along $x^n$. Send $n\uparrow\infty$.  \end{proof}

\begin{Lem}\label{lem:alloc_integrable}
Let $\theta$ be an allocation. If $\frac{\theta}{x}$ is integrable, then $\theta\in \Theta$, i.e. $\theta$ is self-financing.
\end{Lem}

\begin{proof}
If $\frac{\theta}{x}$ is integrable (Def.~\ref{def:pathwise}), then the first term of (\ref{eq:ito}) exists and is continuous. 
The second term also exists and continuous due to \cite[Lem. 4.15]{CC2}. We follow the approach in \cite[Thm 5.14 (A.3)]{CC2} and apply a second order Taylor expansion to each individual term of the following finite sum \begin{eqnarray*}
F(t,x^n_t):=\sum_{\pi_n\ni t_i\leq t}\ln\left(1+\theta(t_i,x^{n}_{t_{i-}})\cdot \frac{\Delta x^{n}(t_i)}{x(t_i)}\right).
\end{eqnarray*}As $n\uparrow\infty$, we observe that the first and second order terms converge to the corresponding first two terms in (\ref{eq:ito}). Since the remainder term $R(t,x^n_t)$ is bounded by (\ref{eq:ito_bound}), we may apply the decomposition technique in \cite[Thm 5.14 (A.6)]{CC2} and conclude that $R(t,x_t):=\lim_{n} R(t,x^n_t)$ exists. Let $G(t,x_t)$ denote (\ref{eq:ito_bound}), we observe that \begin{eqnarray*}
|R(t,x_t)-R(s,x_s)|\leq |G(t,x_t))-G(s,x_s)|,
\end{eqnarray*}for all $0<s<t$ and $x\in\Omega$. Since $G$ is continuous \cite[Lem. 4.15]{CC2}, the continuity of $R$ follows. Having established that $F(t,x_t)$ exists and is continuous, the proof is complete by an application of Thm.~\ref{thm:exist} to $\E:=\exp(F)$.
\end{proof}

\begin{Eg}[Softmax is self-financing]\label{cor:softmax}
\begin{eqnarray}\label{eq:softmax}
\theta(t,x_t):=\frac{1}{\sum_{i=1}^{d}e^{\int_0^{t}\frac{1}{x_i}dx_i}}\left(e^{\int_0^{t}\frac{1}{x_1}dx_1},\ldots,e^{\int_0^{t}\frac{1}{x_d}dx_d}\right)'\in \Theta.
\end{eqnarray}
\end{Eg}

\begin{proof}
The functional \begin{eqnarray*}
F(t,x_t):=\ln\left(\sum_{i=1}^{d}\exp\left(\int_0^{t}\frac{1}{x_i}dx_i\right)\right)
\end{eqnarray*}is $C^{1,2}$ in the sense of \cite[Def. 4.11]{CC2} and hence by \cite[Prop. 5.15]{CC2}  $\nabla_{x}F$ is integrable and \begin{eqnarray*}
\nabla_{x}F_{i}(t,x_{t})=\frac{\theta_i(t,x_t)}{x_i(t-)}
\end{eqnarray*}for all $1\leq i\leq d$, hence $\frac{\theta}{x}$ is integrable. By Lem.~\ref{lem:alloc_integrable}, the proof is complete. 
\end{proof}

\begin{Rem}
The allocation strategy (\ref{eq:softmax}) is path-dependent. The map $x\mapsto \theta(t,x_t)$ is not uniform continuous on $\Omega$ and the path $t\mapsto \theta(t,x_t)$ is not necessarily continuous or of finite variation. 
\end{Rem}

\begin{Cor}[Convex hull is self-financing]\label{cor:hull}
\begin{eqnarray*}
\B\subset \Theta.
\end{eqnarray*}
\end{Cor}

\begin{proof}
Let $\theta(b)\in \B$. Since each $\theta^{(k)}\in\Theta$, we can apply Lem~\ref{lem:ito} to each $\theta^{(k)}$ and deduce that each $\frac{\theta^{(k)}}{x}$ is integrable, i.e. $\int \frac{\theta^{(k)}}{x} dx\in\M$ exists. Since $\M$ is a vector space, we conclude
that $\theta(b)$ is integrable. The proof is complete
by Lem.~\ref{lem:alloc_integrable}.
\end{proof}

\begin{Rem}\label{rem:open}
For a small enough $\epsilon>0$, let \begin{eqnarray*}
b\in\Delta_{m}^{\epsilon}:=\left\{b\in\R^{m}\Bigg{|} b_{k}\in(-\epsilon,1+\epsilon); \sum_{k=1}^{m}b_k<1+\epsilon\right\}\supset\Delta_{m}.
\end{eqnarray*}Although $\theta(b)$ will cease be a long-only allocation, we remark from the lines of proof in Lem.~\ref{lem:alloc_integrable} and Cor.~\ref{cor:hull} that $\ln W(\theta(b))$ can still be defined by the RHS of (\ref{eq:ito}) due to the constraint imposed by (\ref{eq:omega}).
\end{Rem}

\begin{Lem}[Regularity]\label{lem:uniform}Let $(t,x_t)\in\Lambda$, then \begin{itemize}
    \item [(i)] $b\mapsto\ln W(b)$ is concave on $\Delta_m$.
    \item [(ii)] $b\mapsto \ln W^n(b)$ is infinitely differentiable on $\Delta_m$.
    \item [(iii)] $\ln W^n(b)\xrightarrow{n}\ln W(b)$ uniformly on $\Delta_m$. 
    \item [(iv)] $b\mapsto \ln W(b)$ is continuous on $\Delta_m$.
    \item [(v)] $W^*:=\underset{b\in \Delta_{m}}{\max}W(b)$ exists
\end{itemize}
\end{Lem}

\begin{proof}
All maps are well defined due to Cor.\ref{cor:hull}. For a fixed $z\in\R^{m}; z>-1$, the function $b\mapsto\ln(1+b\cdot z)$ is concave on $\Delta_{m}$. Since $b\mapsto (\ln W(b)(t,x_t))$ is the limit (\ref{eq:dd}) of a sum of concave functions, we obtain (i). (ii) is due to the fact that $b\mapsto\ln W^n(b)$ is a finite sum of infinitely differentiable functions on $\Delta_{m}$. For (iii), we first observe that for a sufficiently small $\epsilon>0$, $b\mapsto \ln W^{n}(b)$ can be extended (Rem.~\ref{rem:open}) to an open convex set $\Delta^{\epsilon}_{m}\supset\Delta_{m}$ and that $b\mapsto \ln W^{n}(b)$ will be finite and concave on $\Delta^{\epsilon}_{m}$ (Hessian is negative semi-definite) . Since  $(\ln W^{n}(b))_n$ converges pointwise on $\Delta^{\epsilon}_{m}$, by \cite[Thm. 10.8]{RF}, (iv) follows. (iv) is due to (ii) and (iii). (v) follows from (i) and (iv).\end{proof}

\begin{Lem}[Integrability]\label{lem:leb_integrable}
 Let $(t,x_t)\in\Lambda$, the following maps:
\begin{itemize}
\item [(i)] $b\mapsto W(b)$,

\item [(ii)] $b\mapsto \theta(b)W(b)$,

\item [(iii)] $b\mapsto \nabla_{x}W(b)$,
\end{itemize}are all integrable on $\Delta_m$.
\end{Lem}

\begin{proof}
All maps are bounded and measurable due to Lem.~\ref{lem:uniform} and (\ref{eq:pde}).
\end{proof}

\begin{Thm}[Universal portfolio on convex hull]\label{thm:dct}
Define
\begin{eqnarray*}
\widehat{\theta}:=\frac{\int_{\Delta_{m}}\theta(b)W(b)db}{\int_{\Delta_{m}}W(b)db},\end{eqnarray*}then 
$\widehat{\theta}\in\Theta$ and 
\begin{eqnarray*}\widehat{W}:=W(\widehat{\theta})=\frac{1}{|\Delta_{m}|}\int_{\Delta_{m}}W(b)db.\end{eqnarray*}
\end{Thm}

\begin{proof}
The integrals are well-defined at every point in $\Lambda$ according to Lem.~\ref{lem:leb_integrable} and the fact that $W(b)$ is strictly positive (\ref{eq:wmap}). We first observe that $\widehat{\theta}_{i}\geq 0$ and that $\sum_{i=1}^{d}\widehat{\theta}_{i}\leq 1$ by the linearity of the Lebesgue integral at every point in $\Lambda$. Let $z_n$ be a sequence of points that converges to $z$ in $\Lambda$. By Lem.~\ref{lem:uniform}(ii), we observe \begin{alignat}{2}
&0\leq (\theta_{i}(b))(z_n)\leq 1,\quad i=1,\ldots,d\label{eq:dtc_B1},\\
&0<W^n(b):=(W(b))(z_n)\leq W^*(z_n)\label{eq:dtc_B2},\\
&0<\frac{W^n(b)}{\int_{\Delta_m}W^n(b)db}\leq\frac{W^*(z_n)}{\int_{\Delta_m}W^n(b)db}\label{eq:dtc_B3},
\end{alignat}where the upper bounds are independent of $b$ and that if $W^n(b)\rightarrow W(b)(z)$, then \begin{eqnarray}
\int_{\Delta_m}W^n(b)db>0\rightarrow\int_{\Delta_m}W(b)(z)db>0\label{eq:dtc_B4}
\end{eqnarray}by the fact that $W^n(b), W(b)(z)$ are strictly positive, (\ref{eq:dtc_B2}) and the generalised dominated convergence theorem (GDCT). Note that GDCT applies because $W^{*}(z_n)\rightarrow W^{*}(z)$ and $\int_{\Delta_m} W^{*}(z_n)db \rightarrow \int_{\Delta_m} W^{*}(z)db$. Using the bounds given by (\ref{eq:dtc_B1}) \& (\ref{eq:dtc_B3}), the convergence in (\ref{eq:dtc_B4}) and the (GDCT), we have established that $\widehat{\theta}$ is regulated (Def.~\ref{def:version}) and that\begin{eqnarray*}M:=\frac{1}{|\Delta_{m}|}\int_{\Delta_{m}}W(b)db\end{eqnarray*} is positive and continuous (Def.~\ref{prop:pi}) on $\Lambda$. 

    Let $(t,x_t)\in\Lambda^{+}$, , $1\leq i \leq d$ and $\epsilon>0$. Since $W(b)$ is a pathwise integral (Def.~\ref{def:pathwise}), we first observe $W(b)(t+h,x_t)-W(b)(t,x_t)=0$ for all $h\geq 0$ and deduce that $\D M(t,x_t)=0$. Also for sufficiently small $e$: \begin{alignat*}{2}
\frac{1}{\epsilon}\left((W(b))(t,x_t+\epsilon e_i\I_{[t,\infty)})-(W(b))(t,x_t)\right)&=\nabla_{x_i}W(b)(t,x_{t-})\\
&=\left(\frac{(\theta_i(b))}{x_i}W(b)\right)(t,x_{t-}),
\end{alignat*}by (\ref{eq:pde}) and that $W(b)$ is a pathwise integral (Def.~\ref{def:pathwise}). Integrating both sides with regard to $b$ for every $1\leq i\leq d$ and send $\epsilon\downarrow 0$, we obtain\begin{alignat*}{2}
\nabla_{x}\left(\int_{\Delta_{m}}W(b)db\right)(t,x_t)&=\left(\int_{\Delta_{m}}\frac{\theta(b)}{x}W(b)db\right) (t,x_{t-}),\\
&=\left(\frac{\widehat{\theta}}{x}\int_{\Delta_{m}}W(b)db\right)(t,x_{t-}),
\end{alignat*}i.e. \begin{eqnarray}
\nabla_{x}M=\left(\frac{\widehat{\theta}}{x}M\right)_{-}.\label{eq:up}
\end{eqnarray}Observe that the left limit of a regulated functional (Def.~\ref{def:version}) is left-continuous and strictly causal. Thus, we have established that $M\in\M_{+}$ according to Def~\ref{def:smooth} and that $\widehat{\theta}\in\Theta$ according to Thm.~\ref{thm:alloc_self}. Since $M_0=\widehat{W}_0=1$, the proof is complete by Cor.~\ref{cor:ddmap}, (\ref{eq:wmap}) and (\ref{eq:up}).
\end{proof}

\Eg[Portfolio construction]\label{Eg:pc}\noindent\\
A combination of finite and universal algorithms may be applied individually as well as \emph{successively} to different groups of portfolio allocation strategies. In this case, the resulting final algorithm becomes meta-learning. 

One example is to apply the universal algorithms to the convex combination of 1. the $\frac{1}{d}$ constant re-balanced portfolio, which may be considered a mean-reverting/martingale strategy i.e. $\theta^{(1)}:\equiv\frac{1}{d}$ and 2. the evenly distributed buy and hold given in example \ref{Eg:single} i.e.:

$$\theta^{(2)}(t,x_t):=\frac{1}{\sum_{i=1}^{d}\frac{x_i(t)}{x_i(0)}}\left(\frac{x_1(t)}{x_1(0)},\ldots,\frac{x_{d}(t)}{x_{d}(0)}\right)'.
$$According to Thm.\ref{thm:AA} the convex combination of strategies may then be defined as follow:\begin{eqnarray*}
\theta(b)(t,x_t)=b_1\left(\frac{1}{d},\ldots,\frac{1}{d}\right)'+\frac{b_2}{\sum_{i=1}^{d}\frac{x_i(t)}{x_i(0)}}\left(\frac{x_1(t)}{x_1(0)},\ldots,\frac{x_{d}(t)}{x_{d}(0)}\right)'
\end{eqnarray*}
for $b\in\Delta_{2}$, which are no longer constant rebalanced allocations.

As noted by Cover \cite[\S 1]{TC}, the exponential growth rate of wealth generated by the universal portfolio may not be significantly better than that generated by the finite algorithm (i.e. evenly distributed buy and hold), when the stocks are positively correlated or some stocks are inactive (i.e. $b^{*}(t)\notin\mathring\Delta_d$). In this case, one may first apply the finite algorithm to each individual sector of stocks and proceed to apply the universal algorithms at the sectorial level: \begin{alignat*}{2} 
\theta^{(1)}(t,x_t)&=\frac{1}{\sum_{i\in K}\frac{x_i(t)}{x_i(0)}}\left(\frac{x_1(t)}{x_1(0)},\ldots,\frac{x_{d_1}(t)}{x_{d_1}(0)}, \mathbf{0}_{1\times d_2}\right)',\\
\theta^{(2)}(t,x_t)&=\frac{1}{\sum_{j\in L}\frac{x_j(t)}{x_j(0)}}\left(\mathbf{0}_{1\times d_1},\frac{x_{d_{1}+1}(t)}{x_{d_{1}+1}(0)},\ldots,\frac{x_{d}(t)}{x_{d}(0)}\right)',
\end{alignat*}where $\{1,\ldots,d\}$ is the disjoint union of $K$ and $L$. According to example \ref{Eg:single} and Thm.\ref{thm:AA}, the convex combination of strategies may then be defined as follow: \begin{eqnarray*}
\theta(b)(t,x_t)=\frac{b_1}{\sum_{i\in K}\frac{x_i(t)}{x_i(0)}}\left(\frac{x_1(t)}{x_1(0)},\ldots,\frac{x_{d_1}(t)}{x_{d_1}(0)}, \mathbf{0}\right)'+\frac{b_2}{\sum_{j\in L}\frac{x_j(t)}{x_j(0)}}\left(\mathbf{0},\frac{x_{d_{1}+1}(t)}{x_{d_{1}+1}(0)},\ldots,\frac{x_{d}(t)}{x_{d}(0)}\right)'
\end{eqnarray*} for $b\in\Delta_{2}$. In both examples, if we denote $$
\hat{b}(t,x_t):=\frac{\int_{\Delta_2}b W_t(\theta(b))db}{\int_{\Delta_2}W_t(\theta(b))db}, 
$$the universal allocation strategies are then:
\begin{eqnarray*}
\hat{\theta}(t,x_t)=\hat{b_1}(t,x_t)\theta^{(1)}(t,x_t)+\hat{b_2}(t,x_t)\theta^{(2)}(t,x_t).
\end{eqnarray*}and the computational problem of an integral over simplex is reduced from dimension $d$ to $2$.


\begin{Eg}[Explicit formulae]\label{Eg:Kelly}\noindent\\
Let $(t,x_t)\in\Lambda_{+}$, $\Sigma_t:=\Sigma(t,x_t)$ be positive definite, $b^{*}(t)$ be a maximiser of $W_t(b)$ over $\Delta_m$ and\begin{eqnarray*}
R_t:=\left(\int_0^{t}\frac{\theta^{(1)}}{x}dx,\ldots,\int_{0}^{t}\frac{\theta^{(m)}}{x}dx\right)'. 
\end{eqnarray*}Suppose $t\longmapsto x(t)$ is continuous and $b^{*}(t)\in\mathring{\Delta}_m$, then:
\begin{itemize}

\item[(i)] $
b^{*}(t)=\Sigma^{-1}_t R_t
.$

\item[(ii)] $\ln W^{*}_t=\frac{1}{2}R_t'\Sigma^{-1}_t R_t.$

\item[(iii)] $\frac{\widehat{W}_t}{W^{*}_t}=\frac{1}{|\Delta_m|}\int_{\Delta_m}\exp\left(-\frac{1}{2}(b-b^{*}(t))'\Sigma_t(b-b^{*}(t))\right)db.$

\item[(iv)] $\widehat{\theta}_t=\sum_{k=1}^{m}\widehat{b}_k(t)\theta_t^{(k)}$, where $ \widehat{b}(t):=\frac{\int_{\Delta_m} b \exp\left(-\frac{1}{2}(b-b^{*}(t))'\Sigma_t(b-b^{*}(t))\right)db}{\int_{\Delta_m}\exp\left(-\frac{1}{2}(b-b^{*}(t))'\Sigma_t(b-b^{*}(t))\right)db}.$
\end{itemize}If in addition,  $\lambda_{\min}(t)\uparrow \infty$ and $b^{*}(t)\rightarrow b^{*}\in \mathring\Delta_m$, then 
$\widehat{b}(t)\rightarrow {b}^*.$

\begin{proof}
By \ito's lemma \ref{lem:ito}, we first differentiate $b\mapsto \ln W_t(b)$ and obtain (i) \& (ii). Using (i) and (ii), we then obtain 
\begin{alignat}{2}
\ln W_t(b)&=b'R_t-\frac{1}{2}b'\Sigma_t b\nonumber\\
&=b'\Sigma_t b^{*}(t)-\frac{1}{2}b'\Sigma_t b\nonumber\\
&= \frac{1}{2}b^{*'}(t)\Sigma_t b^{*}(t)-\frac{1}{2} (b-b^{*}(t))'\Sigma_t(b-b^{*}(t))\nonumber\\
&=\ln W^{*}_t-\frac{1}{2} (b-b^{*}(t))'\Sigma_t(b-b^{*}(t)).\label{eq:gaussian}
\end{alignat}Taking the exponent of (\ref{eq:gaussian}), we obtain (iii) \& (iv) from Thm.~\ref{thm:dct}. Thus, $\widehat{b}(t)$ is the expected value of a Gaussian distribution with mean $b^{*}(t)$ and covariance $\Sigma^{-1}_t$ conditioning on $\Delta_m$.
\end{proof}

\end{Eg}

\begin{Rem}
All quantities in Eg.~\ref{Eg:Kelly} are in general path-dependent. The asymptotic limits of $\frac{1}{t}R_t$ and $\frac{1}{t}\Sigma_t$ are not required to exist. 
\end{Rem}

\begin{Rem}
The map $b\mapsto W_t(b)$ is proportional to a Gaussian density.
\end{Rem}

\begin{Rem}
If $\theta^{(k)}=e_k$ for $k=1,\ldots m=d$ (standard basis), $t\mapsto R_t$ is a sample path of a Brownian Motion and $b^{*}(t)\rightarrow b^{*}\in \mathring\Delta_m$, then the universal portfolio $\widehat{\theta}_t$ converges to Kelly's criterion \cite{ET}.
\end{Rem}

\begin{Thm}[The exact ratio]\label{Thm:universal}\noindent\\
Let $(T,x_T)\in\Lambda_{+}$ and $\Sigma_T:=\Sigma(T,x_T)$ be positive definite. Then $W_T(b)$ has a unique maximum $b^{*}$ in $\Delta_{m}$. If the maximum lies in the interior of $\Delta_{m}$, then there exists $\delta\in[-\delta_{+},\delta_{-}]$;

\begin{eqnarray}
\frac{\widehat{W}_T}{W^*_T}= 
\mu_{m}\left(\frac{\Sigma^{1/2}_{T}(\Delta_m -b^*)}{(1-\delta)}\right)\frac{(1-\delta)^{m}m!(2\pi)^{m/2}}{\sqrt{\det\Sigma_{T}}},
\label{eq:u_bound}\end{eqnarray}where $\mu_m$ is the standard Gaussian measure on $\R^{m}.$ If $t\mapsto x(t)$ is continuous, then $\delta=0$.
\end{Thm}

\begin{Rem}
The Gaussian factor in \eqref{eq:u_bound} may be written as
\[
\mu_{m}\!\left(\tfrac{\Sigma^{1/2}_{T}(\Delta_m - b^*)}{1-\delta}\right)
= \gamma\!\big(b^*, (1-\delta)^2\Sigma_T^{-1}\big)(\Delta_m),
\]
i.e. the Gaussian measure of $\Delta_m$ with mean $b^*$ and covariance  $(1-\delta)^2\Sigma_T^{-1}$.
\end{Rem}

\begin{proof}
By Lem.~\ref{lem:uniform}, we can extract a sequence of maximisers $(b^*_n)_n\subset\Delta_{m}$, each maximising $b\mapsto \ln W^n_T(b)$ over $\Delta_m$. Since $\Delta_m$ is compact, we can assume (for ease of notation without passing to a subsequence) that there exists a $b^{*}\in\Delta_m$; 
$b^*_n\xrightarrow{n} b^{*}$. By the convexity of $\Delta_m$, Lem.~\ref{lem:uniform} and \cite[Thm. 2.1]{UCP}, it follows that $b^{*}$ is a maximiser for $\ln W_T(b)$ and \begin{eqnarray*}
\ln W^n_T(b^*_n)\xrightarrow{n}\ln W^*_T.
\end{eqnarray*}By Lem.~\ref{lem:uniform}(ii) and a second order Taylor approximation, we then expand \begin{eqnarray}\ln W^n_T(b)-\ln W^n_T(b^*_n)=\nabla_{b}\ln W^n_T(b^*_n)\cdot (b-b^*_n)+\frac{1}{2}
\left\langle\nabla^2_{b}\ln W^n_T(\tilde{b}_n), (b-b^*_n)(b-b^*_n)'\right\rangle,\label{eq:u2}\end{eqnarray}where $\tilde{b}_n:=\alpha_n (b-b^*_n) +b^*_n\in\Delta_m$, $\alpha_n\in (0,1)$. Since $b^*_n$ is a maximiser of $\ln W_T^n(b)$, we observe that the first order term either vanishes (if $b^*_n$ is at the interior of $\Delta_m$) or becomes non-positive (if $b^*_n$ is at the boundary of $\Delta_m$). Hence, we have first established that \begin{eqnarray}
a_n:=\nabla_{b}\ln W^n_T(b^*_n)\cdot (b-b^*_n)\leq 0.\label{eq:foc}  
\end{eqnarray}For the second order term, we compute the Hessian and obtain for each $1\leq k,l\leq m$,
\begin{eqnarray}
-\left(\nabla^{2}_{b}\ln {W}^{n}_{T}(\tilde{b}_{n})\right)_{k,l}=\sum_{\pi_{n}\ni t_{i}\leq T}
\frac{\left\langle \theta_{-}^{(k)}\theta_{-}^{(l)}{'},\frac{\Delta {x}^{n}(t_i)}{x(t_i)}\frac{\Delta {x}^{n}(t_i){'}}{x(t_i)}\right\rangle}{\left(1+\theta(\tilde{b}_{n})_{-}\cdot \frac{\Delta x^{n}(t_i)}{x(t_{i})}\right)^{2}}\label{eq:u3}.
\end{eqnarray}Since each\begin{eqnarray}\left(\left\langle \theta_{-}^{(k)}\theta_{-}^{(l)}{'},\frac{\Delta {x}^{n}(t_i)}{x(t_i)}\frac{\Delta {x}^{n}(t_i){'}}{x(t_i)}\right\rangle\right)_{1\leq k,l\leq m}\label{eq:semi}\end{eqnarray} is positive semi-definite (we refer to proof in Lem.~\ref{lem:definite}) and that (\ref{eq:omega}) implies\begin{eqnarray}
(1+\delta_{+})^{2}>\left(1+\theta(\tilde{b}_n)_{-}\cdot \frac{\Delta x^n(t_i) }{x(t_{i})}\right)^{2}> (1-\delta_{-})^{2}>0,\label{eq:u4}
\end{eqnarray}it follows from (\ref{eq:u2}), (\ref{eq:foc}) (\ref{eq:u3}), (\ref{eq:semi}) \& (\ref{eq:u4}) for sufficiently large $n$,
\begin{alignat*}{2}
\frac{1}{2(1-\delta_{-})^2}\sum_{k,l\leq m}\sum_{\pi_n \ni t_i\leq T}&(b-b^*_n)_k(b-b^*_n)_l\left\langle \theta_{-}^{(k)}\theta^{(l)}_{-}{'},
\frac{\Delta x^{n}(t_i)}{x(t_i)}\frac{\Delta x^{n}(t_i){'}}{x(t_i)}\right\rangle-a_n\\
&\geq\ln \frac{W^n_T(b^*_n)}{W^n_T(b)}\geq\\
\frac{1}{2(1+\delta_{+})^2}\sum_{k,l\leq m}\sum_{\pi_n \ni t_i\leq T}&(b-b^*_n)_k(b-b^*_n)_l\left\langle \theta_{-}^{(k)}\theta^{(l)}_{-}{'},
\frac{\Delta x^{n}(t_i)}{x(t_i)}\frac{\Delta x^{n}(t_i){'}}{x(t_i)}\right\rangle.\end{alignat*}By \cite[Lem. 4.15]{CC2}, the fact that $b^*_{n}\rightarrow b^{*}$ and the triangle inequality, we can send $n\uparrow\infty$ and establish that \begin{eqnarray}
\frac{1}{2(1+\delta_{+})^2}(b-b^{*})'\Sigma_T(b-b^{*})\leq\ln \frac{W^{*}_T}{W_T(b)}\leq  \frac{1}{2(1-\delta_{-})^2}(b-b^{*})'\Sigma_T(b-b^{*})+\limsup_n |a_n|
\label{eq:u5}\end{eqnarray}Since $\Sigma_T$ is positive definite by assumption, we first conclude from the LHS of
(\ref{eq:u5}) that the maximiser $b^{*}$ is unique. If $b^{*}\in\mathring\Delta_{m}$, it follows from (\ref{eq:foc}) that $\lim a_n=0$. Multiplying both sides of (\ref{eq:u5}) by $-1$, exponentiating and passing through $\mint_{\Delta_m}:=\frac{1}{|\Delta_m|}\int_{\Delta_{m}}$, we obtain by Thm.~\ref{thm:dct}:
\begin{alignat*}{2}
&\mint_{\Delta_m}\exp\left( -\frac{(b-b^*)'\Sigma_{T}(b-b^*)}{2(1-\delta_-)^2}\right)db
\leq \frac{\widehat{W}_T}{W_T^{*}}\leq  \mint_{\Delta_m}\exp\left( -\frac{(b-b^*)'\Sigma_{T}(b-b^*)}{2(1+\delta_+)^2}\right)db.\\
\end{alignat*}Since $
\delta\longmapsto \mint_{\Delta_m}\exp\left( -\frac{(b-b^*)'\Sigma_{T}(b-b^*)}{2(1-\delta)^2}\right)db
$ is continuous in $[-\delta_+,\delta_-]$, by an application of the intermediate value theorem
, we obtain (\ref{eq:u_bound}). The last statement follows from Eg.~\ref{Eg:Kelly}(iii).
\end{proof}

\begin{Cor}[Asymptotics]\label{cor:asym}
Let $x\in\Omega$, $b^{*}(t)$ a maximiser of $W_t(b)$ and $\lambda_{\min}(t)>0$ for a $t>0$. Suppose $b^{*}(t)\rightarrow b^*\in\mathring\Delta_m$, then the following hold:

\begin{itemize}

    \item [(i)] $\lambda_{\min}(t)\uparrow\infty$ implies that

\begin{eqnarray*}
\frac{W^*_t}{\widehat{W}_t}\sim\left(
\frac{\sqrt{\det\Sigma_{t}}}{(1-\delta_{t})^mm!(2\pi)^{m/2}}\right)
\label{eq:asym}\end{eqnarray*} in the sense that the ratio of both sides converges to $1$, where $\delta_{t}\in[-\delta_{+},\delta_{-}]$. If in addition, $t\mapsto x(t)$ is continuous, then $\delta_{t}$ is identically 0.

\item [(ii)] In general,\begin{eqnarray*}\frac{1}{t}\ln\frac{W^*_t}{\widehat{W}_t}=\frac{O\left(\ln \lambda_{\max}(t)\right)}{t}, \qquad
\frac{1}{t}\ln\frac{W^*_t}{\widehat{W}_t}=\frac{\mathit{\Omega}\left(\ln \lambda_{\min}(t)\right)}{t}.\end{eqnarray*}In particular, if $\lambda_{max}(t)$ has at most polynomial growth, then 
$$\frac{1}{t}\ln\frac{W^*_t}{\widehat{W}_t}=\frac{O\left(\ln t\right)}{t}\longrightarrow 0.$$

\end{itemize}
\end{Cor}

\begin{proof}
Let $\gamma(b,C)$ denotes the $m$-dimensional
Gaussian measure centered at $b\in\Delta_m$ with covariance $C$. 
We first observe \begin{eqnarray*}
\gamma(b^*(t),C_t)(\Delta_m)=\mu_m\left(\frac{\Sigma^{1/2}_{t}(\Delta_m -b^*(t))}{(1-\delta_{t})}\right)\geq \mu_m\left(\frac{\Sigma^{1/2}_{t}(\Delta_m -b^*(t))}{(1+\delta_{+})}\right)=\gamma(b^*(t),C^{+}_t)(\Delta_m),
\end{eqnarray*}where 
$C_t:=(1-\delta_{t})^2\Sigma_t^{-1}<(1+\delta_{+})^2\Sigma_t^{-1}=:C^{+}_t$ in the sense of Loewner order. If
$b^{*}(t)\rightarrow b^*\in\mathring\Delta_m$ and $\lambda_{\min}(t)>0$ for some $t_0>0$, then for every sufficiently small $\epsilon>0$, there exists $t_1>t_0$;
$$
\gamma(b^*(t),C_t)(\Delta_m)\geq\gamma(b^*(t),C^+_t)(\Delta_m)\geq \gamma(b^*(t),C^{+}_{t_0})(\Delta_m)\geq \gamma(b^*,C^{+}_{t_0})(\Delta_m)-\frac{\epsilon}{2}>0
$$
for all $t\geq t_1$ because $t\mapsto C^{+}_t$ is monotonic decreasing in the sense of Loewner order (Rem.~\ref{rem:definite}). We have first established that $\gamma(b^*(t),C_t)(\Delta_m)$ is bounded away from $0$. If in addition, $\lambda_{\min}(t)\uparrow\infty$, then for every $\frac{\epsilon}{2}>0$, there exists a $t_0>0$ such that for all $t\geq t_0$, it holds:$$
\gamma(b^*,C_t)(\Delta_m)\geq \gamma(b^*,C^{+}_t)(\Delta_m)\geq\gamma(b^*,C^{+}_{t_0})(\Delta_m)\geq  \I_{\{b^{*}\}}(\Delta_m)-\frac{\epsilon}{2}=1-\frac{\epsilon}{2}. 
$$Hence, we combine with the first inequality i.e. there exists $t_1>t_0$; 
$$
\gamma(b^*(t),C_t)(\Delta_m)\geq\gamma(b^*,C^{+}_{t_0})(\Delta_m)-\frac{\epsilon}{2}\geq 1-\epsilon
$$for all $t\geq t_1$ and establish that 
$$
\gamma(b^*(t),C_t)(\Delta_m)\uparrow 1.
$$By Thm.~\ref{Thm:universal}, we have obtained (i). Since $(\lambda_{\min}(t))^m\leq \det \Sigma_t \leq (\lambda_{\max}(t))^m$, we obtained (ii).

\end{proof}

\COI{None declared.}\\
\DA{Not Applicable.}\\
\ACKNO{Part of this research was conducted at Imperial College London and was supported by UKRI-EPSRC grants under project references EP/W007215/1 and EP/W522673/1.}

\bibliographystyle{unsrt}
\bibliography{uni}

\begin{thebibliography}{25}



\bibitem{RM}
Merton, R.C. (1969)
\newblock {\em Lifetime Portfolio Selection under Uncertainty: The Continuous-Time Case}.
\newblock The Review of Economics and Statistics, 51(3): 247-257.

\bibitem{HF}
\follmer , H. (1981)
\newblock {\em Calcul d'Ito sans probabiliti{\'e}s}.
\newblock S{\'e}minaire de probabilit{\'e}s (Strasbourg), 15:143-150.

\bibitem{UCP}
Kanniappan, P., Sastry, S. (1983)
\newblock {\em Uniform Convergence of Convex Optimization Problems}. 
\newblock Journal of Mathematical Analysis and Applications, 96, 1-12.

\bibitem{VV}
Vovk, V. (1990) 
\newblock {\em Aggregating strategies}.
\newblock Proc. 3rd Annual Workshop on Computational Learning Theory, 371-383.

\bibitem{TC}
Cover, T. (1991) 
\newblock {\em Universal Portfolio}.
\newblock Mathematical Finance, 1(1): 1-29.

\bibitem{ET}
Thorp, E. O., Rotando, L. M. (1992)
\newblock {\em The Kelly criterion and the Stock Market}.
\newblock The American Mathematical Monthly. 99 (10): 922-931.

\bibitem{FJ}
Jamshidian, F. (1992) 
\newblock {\em Asymptotically Optimal Portfolios}. 
\newblock Mathematical Finance, 2(2): 131-150. 

\bibitem{TL}
Lyons, T. (1995)
\newblock {\em Uncertain volatility and the risk-free synthesis of derivatives}.
\newblock Applied Mathematical Finance, 2(2): 117-133.

\bibitem{CO}
Cover, T., Ordentlich, E. (1996) 
\newblock {\em Universal Portfolios with Side Information}. 
\newblock IEEE Transactions on Information Theory, 42(2):348-363.

\bibitem{RF}
Rockafellar, R.T. (1997) 
\newblock {\em Convex Analysis}. 
\newblock Princeton University Press.  

\bibitem{HI}
Ishijima, H. (2002)
\newblock {\em Bayesian Interpretation of continuous-time Universal Portfolios}.
\newblock Journal of the Operations Research, Society of Japan 45 (4): 362-372.

\bibitem{ERF}
Fernholz, E.R. (2002). 
\newblock {\em Stochastic Portfolio Theory}
\newblock Springer, New York, NY.

\bibitem{YK}
Kalnishkan, Y. (2009) 
\newblock {\em The Aggregating Algorithm as Laissez-Faire Investment}
\newblock Technical Report, Computer Learning Research Centre, Royal Holloway London, TR-09-02. 

\bibitem{CF}
Cont, R. and D. Fourni\'e (2010) 
\newblock {\em Change of variable formulas for non-anticipative functionals on path space}, 
\newblock Journal of Functional Analysis, 259:1043-1072.


\bibitem{AS1}
Schied, A. (2014)
\newblock {\em Model-free CPPI}.
\newblock Journal of Economic Dynamics and Control 40: 84-94,

\bibitem{AS2}
Schied, A., Speiser, L., Voloshchenko, I. (2018)
\newblock {\em Model-free portfolio theory and its functional master formula}.
\newblock SIAM Journal on Financial Mathematics 9 (3): 1074-1101.

\bibitem{CC}
Chiu, H., Cont, R. (2018)
\newblock {\em On pathwise quadratic variation for cadlag functions}.
\newblock Electronic Communications in Probability, 85: 1-12.

\bibitem{CSW}
Cuchiero, C., Schachermayer, W., Wong, L. (2019)
\newblock {\em Cover's universal portfolio, stochastic portfolio theory, and the num\'{e}raire portfolio}.
\newblock Mathematical Finance 29 (3): 773-803.

\bibitem{BD}
Dupire, B. (2019) 
\newblock {\em Functional \ito calculus}.
\newblock Quantitative Finance, 19: 721-729.

\bibitem{KK}
Karatzas, I., Kim, D. (2020) 
\newblock {\em Trading Strategies Generated Pathwise by Functions of Market Weights}.
\newblock Finance \& Stochastics, 24: 423-463.

\bibitem{CC2}
Chiu, H., Cont, R. (2022)
\newblock {\em Causal functional calculus}.
\newblock Transactions of the London Mathematical Society, 9(1): 237-269. 

\bibitem{CC3}
Chiu, H., Cont, R. (2023)
\newblock {\em A model-free approach to continuous-time finance}.
\newblock Mathematical Finance, 33(2): 257-273.

\bibitem{All}
Allan, A., Cuchiero, C., Liu, C., Pr\"{o}mel, DJ. (2023)
\newblock {\em Model-free portfolio theory: A rough path approach}.
\newblock Mathematical Finance 33 (3): 709-765.

\bibitem{HS}
Han, X., Schied, A. (2025)
\newblock {\em Universal portfolios in continuous time: an approach in pathwise \ito calculus}.
\newblock arXiv:2504.11881 

\end{thebibliography}

\end{document}